\documentclass[12pt]{article}

\RequirePackage{amsthm,amsmath,amsfonts,amssymb}
\RequirePackage{graphicx}%% uncomment this for including figures
\usepackage{algorithm}
\usepackage{algpseudocode}
\usepackage{booktabs}
\usepackage{mathtools}
\usepackage{multirow}
\usepackage{tikz}
\newcommand*\circled[1]{\tikz[baseline=(char.base)]{
            \node[shape=circle,draw,inner sep=1pt] (char) {#1};}}
\usepackage{fullpage}
\usepackage{enumerate}
\usepackage{natbib}
\usepackage{url} % not crucial - just used below for the URL 
\usepackage{bbm}
\usepackage{caption,subcaption,placeins,color,dsfont,xcolor}
\usepackage{fancyhdr}
\usepackage{hyperref}
\usepackage{accents}
\usepackage{mathtools}
\theoremstyle{plain}

\newtheorem{theorem}{Theorem}[section]
\newtheorem{lemma}[theorem]{Lemma}
\theoremstyle{remark}
\newtheorem{definition}[theorem]{Definition}
\newtheorem{assumption}[theorem]{Assumption}
\newtheorem{proposition}[theorem]{Proposition}
\newtheorem{example}{Example}

\begin{document}

\def\spacingset#1{\renewcommand{\baselinestretch}%
{#1}\small\normalsize} \spacingset{0}

%%%%%%%%%%%%%%%%%%%%%%%%%%%%%%%%%%%%%%%%%%%%%%%%%%%%%%%%%%%%%%%%%%%%%%%%%%%%%%

\title{\bf Model-Based Regression Adjustment with Model-Free Covariates for Network Interference}
\author{Kevin Han \vspace{.2cm} \hspace{.2cm}\\
    Department of Statistics, Stanford University \vspace{.2cm}\\
    and \vspace{.2cm} \\
    Johan Ugander \vspace{.2cm}\\ 
    Department of Management Science and Engineering, Stanford University}
\maketitle

\noindent%

  %\subtitle{...}

\begin{abstract}
{When estimating a Global Average Treatment Effect (GATE) under network interference, units can have widely different relationships to the treatment depending on a combination of the structure of their network neighborhood, the structure of the interference mechanism, and how the treatment
was distributed in their neighborhood. In this work, we introduce a sequential
procedure to generate and select graph- and treatment-based covariates for
GATE estimation under regression adjustment. We show that it is possible to simultaneously achieve
low bias and considerably reduce variance with such a procedure. 
%Under mild assumptions, our regression adjustment procedure leads to a consistent
%estimator, even with these assignment-dependent adjustment covariates. 
To tackle inferential complications caused by our feature generation and selection process, we introduce a way to construct confidence intervals based on
a block bootstrap. We illustrate that our selection procedure and subsequent
estimator can achieve good performance in terms of root mean squared error
in several semi-synthetic experiments with Bernoulli designs, comparing favorably to an oracle estimator that takes advantage of regression adjustments
for the known underlying interference structure.
%as well as a (fractional-exposure-) Horvitz-Thompson estimator. 
We apply our method to a real world experimental dataset with strong evidence of
interference and demonstrate that it can estimate the GATE reasonably well
without knowing the interference process a priori.}
\end{abstract}

\newpage
\spacingset{1}

\section{Introduction}
In standard experiments, researchers typically assume that one unit’s assignment does not affect another unit’s response; this is usually referred to as the assumption of no interference assumption \cite[Chapter~2]{cox58} or the stable unit treatment value assumption (SUTVA) \citep{rubin1974estimating}. However, when experimental units interact with each other, SUTVA is often untenable. Violation of SUTVA has been found in many applications, including politics \citep{sinclair2012}, education \citep{hong2006,rosenbaum2007}, economics \citep{sobel2006,manski2013}, and public health \citep{halloran1995}. Recently, technology companies developing products with social or market interactions have developed methods to manage the considerable interference in their product experiments \citep{EcklesKarrerUgander+2017,pouget-abadie19,karrer2021}. In practice, researchers look for an underlying structure that limits the scope of interference and estimation of causal effects proceeds from assuming the structure. \cite{aronow2017} propose to use a lower dimensional representation of the interference mechanism and estimate causal effects accordingly. In the no-interference literature, regression adjustment has shown to be effective in both theory \citep{lin2013} and practice \citep{cuped}. \cite{Chin+2019} considers regression adjustment under interference when assuming a linear model for the outcomes, and estimate the parameters of the model from the experimental data. Such a linear model assumption is not uncommon and has also been studied in design of experiments \citep{harshawsavje2022} and interference detection \citep{pouget-abadie19}. There has also been literature on new designs that tackle the complication of interference. For example, \cite{UganderKBK13} and \cite{uganderyin2020} consider (randomized) cluster randomized designs that effectively account for interference by doing randomization on cluster level instead of unit level.

In this article, we provide a procedure to estimate the global average treatment effect by using regression adjustment without assuming the true set of features as in \cite{Chin+2019}. We generate the features for adjustment based on observed experimental data in a model-free manner. As an outline for this work, we first give preliminaries of the problem setup and motivate our method through a study of the classic linear-in-means model in econometrics. We then provide our general procedure to generate model-free covariates based on the observed experimental data. Finally, we show how to do estimation and inference for the global average treatment effect with model-free covariates. We conclude with simulations, an empirical applications, and a discussion.

\section{Setup}
Consider a randomized experiment on $n$ units where these is a simple undirected graph $G = (V, \mathcal{E})$ that describes the social network of interactions among $n$ units. The graph $G$ is associated with a symmetric matrix $A \in \mathbb{R}^n$ so that $A_{ij} = 1$ if $(i, j) \in \mathcal{E}$ and zero otherwise. Let $\mathcal{N}_i^{(k)}$ denote the $k$-hop neighborhood around each node $i \in V$. We omit the superscript when $k = 1$ and let $d_i$ denote the degree of each node (or equivalently, $d_i = |\mathcal{N}_i|)$. We denote by $W_i$ the random assignment and $x_i \in \mathcal{X}$ the pre-treatment covariates for unit $i$. We assume that the experimental population is the population of interest and hence view pre-treatment covariates as fixed. We only consider binary treatments but note that extensions to non-binary treatments are straightforward. Throughout, we use lower case letters with the appropriate
subscript for realizations of the random variables and for non-random quantities.

We work under the Rubin causal model \citep{rubin1974estimating, holland1986, imbens_rubin_2015}. For every unit $i$, we associate it with potential outcomes $Y_i(w) \in \mathbb{R}$ for $w \in \{0, 1\}^n$. We are interested in the following causal estimand that we call the Global Average Treatment Effect (GATE):
\begin{equation}
    \tau = \frac{1}{n}\sum_{i = 1}^{n}\mathbb{E}[Y_i(\mathbf{1}) - Y_i(\mathbf{0})]. 
    \label{eq: gate}
\end{equation}
Here $\mathbf{1}$ denotes the $n$-dimensional ones vector and similarly for $\mathbf{0}$. The GATE estimand, also known as the Total Treatment Effect (TTE) in some work \citep{yu2022}, measures the overall effect of the intervention on the experimental units. Under SUTVA, the assignments of other units won't affect one's response and hence there are only two potential outcomes per unit, $Y_i(0)$ and $Y_i(1)$. Under SUTVA, the GATE is then simply the average treatment effect (ATE). 
When there is interference along a network, there may be up to $2^n$ different potential outcomes per unit. In the absence of further assumptions, it is impossible to observe $Y_i(\mathbf{1})$ for some unit $i$ and also observe $Y_j(\mathbf{0})$ for any other unit $j$.

In this work we take a regression perspective and assume two functions $f_0$ and $f_1$ such that for each unit $i$ and each assignment vector $w \in \{0, 1\}^n$,
\begin{equation}
    Y_i(w) = w_if_1(i, w, x_i, G) + (1 - w_i)f_0(i, w, x_i, G) + \epsilon_i, 
    \label{eq: model}
\end{equation}
with $\epsilon_i$'s being exogenous, i.e.\ $\mathbb{E}[\epsilon_i | w] = 0$. The functions $f_0$ and $f_1$ each take as input the node label $i$, the assignment vector $w$, the covariate vector $x_i$ and graph $G$. This approach uses exposure mappings \citep{aronow2017} as functions that map an assignment
vector $w$ and $x_i$ to a specific exposure value so that if two assignment vectors $w$ and $w'$ induce the same exposure value for a unit then they have the same value of potential outcome. Since the potential outcomes only depend on the exposure values, we can view them as a function of exposure values and we can rewrite the potential outcomes as in \eqref{eq: model}. 
%In fact, exposure mappings introduced by \cite{aronow2017} can be thought as a special case of such structural models when we ignore the covariate vector $x_i$. Such structural model is functionally equivalent to the exposure mappings introduced by \cite{aronow2017}. \jucom{Elaborate on this ``equivalence''? Are the two approaches making the same restrictive assumption? Let's discuss briefly.} 
Given \eqref{eq: model}, since functions $f_1$ and $f_0$ are shared across all units, we can use the treated units to estimate $f_1$ and control units to estimate $f_0$. Suppose $\hat{f}_0$ and $\hat{f}_1$ are two estimates of $f_0$ and $f_1$ respectively, then a natural estimator of the GATE would be
\begin{equation*}
    \hat{\tau} = \frac{1}{n}\sum_{i = 1}^{n}[\hat{f}_1(i, \mathbf{1}, x_i, G) - \hat{f}_0(i, \mathbf{0}, x_i, G)].
\end{equation*} 
Unfortunately, estimation of the GATE will be impossible without any further assumptions on the structure of the functions $f_0$ and $f_1$\footnote{\cite{basse-airoldi} has a discussion from an inference perspective.}. To motivate our structural assumptions on $f_0$ and $f_1$, we look at the following example.
\begin{example} [Linear-in-means model] \label{ex: linear-in-means}
Consider the structural model \citep{manski1993, 10.7551/mitpress/6294.003.0005, BRAMOULLE200941}
\begin{equation}
    \mathbf{y} = \alpha\mathbf{1} + \beta\tilde{A}\mathbf{y} + \gamma\mathbf{w} + \delta \tilde{A}\mathbf{w} + \boldsymbol{\epsilon}, \quad \mathbb{E}[\boldsymbol{\epsilon} | \mathbf{w}] = 0, \label{eq: lim}
\end{equation}
where $\mathbf{y}$ is the $n \times 1$ outcome vector, $\tilde{A}$ is the degree-normalized adjacency matrix, i.e., $\tilde{A}_{ij} = A_{ij}/d_i$, $\mathbf{w}$ is the assignment vector, and $(\alpha, \beta, \gamma, \delta)$ are parameters. \cite{BRAMOULLE200941} show that under some mild conditions on the coefficients and the graph $G$, we can rewrite the above model as
\begin{equation}
    \mathbf{y} = \alpha/(1 - \beta)\mathbf{1} + \gamma\mathbf{w} + (\gamma\beta + \delta)\sum_{j = 0}^{\infty}\beta^j\tilde{A}^{j + 1}\mathbf{w} + \sum_{j = 0}^{\infty}\beta^j\tilde{A}^{j + 1}\boldsymbol{\epsilon}. \label{eq: lim_expansion}
\end{equation}
Note that now the outcome is linear in the assignment vector $\mathbf{w}$ as well as $\{\tilde{A}^{j + 1}\mathbf{w}\}_{j = 0}^{\infty}$. Let $f_0(i, w, x_i, G) = f_1(i, w, x_i, G) = \alpha/(1 - \beta) + \gamma w_i + (\gamma\beta + \delta)\sum_{j = 0}^{\infty}\beta^j\tilde{A}^{j + 1}w$ and notice that $\mathbb{E}[\sum_{j = 0}^{\infty}\beta^j\tilde{A}^{j + 1}\boldsymbol{\epsilon} | w] = 0$. Thus, the linear-in-means model \eqref{eq: lim} can be written in the form of \eqref{eq: model}.
\end{example}

While in this example the linear model is infinite-dimensional, the linear structure of \eqref{eq: lim_expansion} motivates us to look at linear models for both $f_0$ and $f_1$. To make it formal, we make the following definition:
\begin{definition} [Linear interference]
We say that the model $\mathcal{Y} = \{Y_i(w): w \in \{0, 1\}^n, i \in [n]\}$ exhibits {\it linear interference} if there exists a function $g: [n] \times \{0, 1\}^n \times \mathcal{X} \times \mathcal{G} \rightarrow \mathbb{R}^K$ and $\theta_0 \in \mathbb{R}^K$, $\theta_1 \in \mathbb{R}^K$ such that $f_0(i, w, x_i, G) = \theta_0^Tg(i, w, x_i, G)$ and $f_1(i, w, x_i, G) = \theta_1^Tg(i, w, x_i, G)$. We call each coordinate function $g_j$ of $g$ a {\it feature} of the interference.
\label{def:linear}
\end{definition}
Despite the simplicity of linear interference, from a graph perspective it can be shown that convolutions on graphs can be well-approximated by linear expansion \citep{HAMMOND2011129}. Such a linear interference assumption is not uncommon \citep{cuped, pouget-abadie19, Chin+2019}. \cite{Chin+2019} shows how to do inference once we have access to the oracle $g$ while \cite{pouget-abadie19} give a testing procedure to detect network interference under linear interference. Moreover, because we are interested in the quality of our estimated functions $\hat{f}_0$ and $\hat{f}_1$ for (only) $w = \mathbf{0}, \mathbf{1}$, we are effectively attempting generalization. Simple models usually generalize well \citep{Bousquet2004, 5955}, and thus linear interference provides credibility of inference without losing flexibility in a world where $g$ can be arbitrarily complex.

Before proceeding, we can simplify \eqref{eq: model} somewhat. Note that
\begin{align}
    Y_i(w) &= w_if_1(i, w, x_i, G) + (1 - w_i)f_0(i, w, x_i, G) + \epsilon_i \notag \\
    &= w_if_1(i, w^{(i \rightarrow 1)}, x_i, G) + (1 - w_i)f_0(i, w^{(i \rightarrow 0)}, x_i, G) + \epsilon_i \notag \\
    &= w_i\tilde{f}_1(i, w^{(-i)}, x_i, G) + (1 - w_i)\tilde{f}_0(i, w^{(-i)}, x_i, G) + \epsilon_i, \label{eq: wlog_form}
\end{align}
where $w^{(i \rightarrow t)}$ denotes the $n-$dimensional vector that replaces $w_i$ by $t$ and $\tilde{f}_t$ is a function of $i, w^{(-i)}$, $x_i$ and $G$ only. Therefore, without loss of generality, we assume that the domain of $g$ and hence the domain of $f_0$ and $f_1$ is $[n] \times \{0, 1\}^{n-1} \times \mathcal{X} \times \mathcal{G}$.

From here on, for presentational simplicity we will omit the pre-treatment covariates $x_i$ in our discussion. Extensions to the case of including pre-treatment covariates will be discussed when not obvious. As a result, $g$ is a function of the node label $i$, the assignment vector $w$ and the graph $G$ only. 

We focus on design that satisfies the following uniformity assumption:
\begin{assumption}[Uniformity]
We assume that $W_i$'s are independent and $\forall i$, $\mathbb{P}(W_i = 1) = p_i$ for some $0 < p_i < 1$. 
\end{assumption}
We make this assumption to follow the common practice of using Bernoulli randomization in network experiments, e.g., \cite{karrer2021}. As an alternative, estimates from designs that accounts for network interference (for example, graph cluster randomization) may suffer from sizable variance \citep{uganderyin2020}. Hereinafter we assume that $W_i$'s are i.i.d.\ Bernoulli$(p)$ random variables with $0 < p < 1$, i.e., we work with data from experiments under a Bernoulli design.

%\jucom{This ``since'' doesn't feel well-connected; lets discuss. I feel this assumption+paragraph is a little out of place} network experiments are usually run by tech companies \cite{karrer2021} and they use Bernoulli randomization as a default choice. On the other hand, estimates from design that accounts for network interference (for example, cluster randomization) may suffer from large variance \cite{uganderyin2020}. %We will discuss the case where $W_i$'s are dependent later. 

If we know the function $g$ a priori, \cite{Chin+2019} provides a complete solution. However, if we don't know the function $g$, then there are three significant challenges, all of which we address in this work. First, how should we \textit{construct} $g$ so that the one we construct approximates the true one? Second, suppose we have many candidate functions then how should we \textit{select} among them? Third, even if we have satisfactory answers to the first two questions, how should we do inference? We will address the first two challenges in the next section and the third challenge later.

\section{Model-free covariates} \label{sec: model-free-covariates}
Now by \eqref{eq: wlog_form}, the function $g$ from Definition~\ref{def:linear} takes node label $i, w^{(-i)}$ and $G$ as input and outputs a $K$-dimensional vector, what $g$ essentially does is to produce $K$ covariates based on $w^{(-i)}$ and $G$ for each unit $i$. In this section, we describe a sequential procedure to generate and select model-free covariates. A high-level description of our method would be that we generate rich candidate features based solely on the graph structure as well as the assignment vector and select among these features based on the observed outcomes. We first give the procedure in Algorithm~\ref{alg: ReFeX-lasso} below and then explain the steps in more detail. We call the procedure ReFeX-LASSO as it builds on the graph mining technique ReFeX \citep{refex} to generate candidate features while using LASSO \citep{lasso} to select features.
\begin{algorithm}
\caption{ReFeX-LASSO} \label{alg: ReFeX-lasso}
\begin{algorithmic}[1]
\Require Graph $G = (V, \mathcal{E})$, assignment vector $w \in \{0, 1\}^n$, maximum number of iterations $T$.
\Ensure A set of covariates $S$.
\State Initialize $S = \{\}$, active feature set $A = \{\}$.
\State For each node/unit $i$, construct $m$ base features and add $m$ base features to $A$.
\For{$t = 1$ to $T$}
\State Regress $y$ on $w$ and features from $S$ and $A$ using LASSO with no penalty on features from $S$.
\State If no feature in $A$ is selected, return $S$. Otherwise, add selected features from $A$ to $S$. 
\State Recursively construct features by performing aggregations of features in $A$ over neighbors in $1$-hop neighborhood.
\State Delete old features in $A$ and add those new features to $A$.
\EndFor
\State Return $S$.
\end{algorithmic}
\end{algorithm}

ReFeX (Recursive Feature eXtraction) was originally designed to generate features for graph mining tasks and can be viewed as a recursive algorithm that starts with base features of each node in the graph and iteratively (i) adds and (ii) prunes features based on aggregations over features from neighboring nodes. ReFeX can be viewed as a simple early precursor to recent methods for graph representation learning based on graph convolution networks (GCNs) \citep{hamilton,kipf2017semisupervised}. We adopt the feature generation step in ReFeX algorithm, but replace the feature pruning part of the original algorithm by LASSO, a modification that allows us to more precisely characterize the features that are available at any given step of the algorithm. 

ReFeX has two ingredients---base features and aggregation functions. Given $w$, $\{x_i\}_{i = 1}^{n}$ and $G$, base features are those features that can be constructed by only looking at each node's 1-hop neighborhood. They can be arbitrary as long as they satisfy this local look-up constraint. Base features can be purely graph features like degree, centrality, clustering coefficient, etc. They can also be pre-treatment covariates $x_i$. Often we would also like to have base features that depend on not just one input of the function $g$ but features computed from two inputs of $g$. For example, features like the number of treated neighbors, which depends on both the assignment vector $w$ as well as the graph $G$. Or the average feature value over all neighbors, which depends on the pre-treatment covariates and $G$. With ReFeX, the base features are chosen by the analyst. Aggregation functions are functions that take features from neighboring nodes as inputs and output a single value. Hence, one aggregation function essentially computes a statistic based on the sample of feature values from neighbors. The aggregation functions again can be arbitrary and chosen by the analyst. Some common examples include min, max, sum, mean and variance \citep{refex}.

We are now ready to introduce the ReFeX-LASSO algorithm. The ReFeX-LASSO algorithm starts with two empty feature sets, the target set $S$ and the active feature set $A$. The first set $S$ stores the selected features and features in $S$ will be used for adjusting the GATE estimate. The active feature set $A$ contains features that were recursively added in the previous step and yet to be selected. At the beginning of the procedure, we construct base features for each unit $i$. Equipped with a set of base features, each time we regress the outcome vector $y$ on features from both set $S$ and set $A$ using LASSO. The LASSO regularization parameter can be chosen by cross-validation and hence we do not need extra hyper-parameters of the algorithm. Note that we do not put a penalty on features in $S$ since they have already been selected and should be kept. The intuition behind this step is that in general features generated later (pulling information from farther in the graph) should not be more predictive than features selected previously. Next, depending on the number of newly selected features, we either terminate the construction and return the current $S$ or add those selected features to $S$ and proceed with the recursive construction. We then need to generate new features and add them to $A$. To do so, we now perform aggregations on old features over all neighboring units. Finally, we add those features to $A$ and delete all old features in $A$.

%\jucom{There are two ``chosen by the analyst'' aspects to ReFeX: (i) the base features and (ii) the aggregation functions. They are discussed in the last two paragraphs, respectively. I think two ``choices'' should be better presented up front before these two paragraphs, in one consolidated place. }

The maximum number of iterations in Algorithm~\ref{alg: ReFeX-lasso} limits the distance in the graph that we can pull information from. Although each step only performs aggregations over neighbors in the 1-hop neighborhood, by repeatedly performing the aggregations we are able to construct features that are informative for the $k$-hop neighborhood. To illustrate this point, we give an example.

\begin{example}[ReFeX and multi-hop information] 
\label{ex: mean-aggregation}
%We show with a concrete example that the recursive step is able to generate features that are not just informative for 1-hop neighborhood. In fact, 
%performing $t$ steps of aggregation on local features can let us obtain features that are informative for $(t+1)$-hop neighborhood. Hence, a common choice for $T$ would be the diameter of the graph. Now, 
Suppose one of the base features we use in ReFeX-LASSO is the fraction of treated neighbors,
\begin{equation*}
\rho_i = \frac{1}{d_i}\sum_{j \in \mathcal{N}_i}w_j,
\end{equation*}
and supposed we limit ourselves to mean aggregation, i.e., we look at each unit's neighbors and aggregate their fraction of treated neighbors using a mean function. We call this new feature $\tilde{\rho}_i$. We then have that
\begin{align*}
    \tilde{\rho}_i &= \frac{1}{d_i}\sum_{j \in \mathcal{N}_i}\rho_j \\
    &= \frac{1}{d_i}\sum_{j \in \mathcal{N}_i}\frac{1}{d_j}\sum_{k \in \mathcal{N}_j}w_k \\
    &= \sum_{j = 1}^{n}\frac{A_{ij}}{d_i}\sum_{k = 1}^{n}\frac{A_{jk}}{d_j}w_k \\
    &= \sum_{j = 1}^{n}\tilde{A}_{ij}\sum_{k = 1}^{n}\tilde{A}_{jk}w_k \\
    &= [\tilde{A}^2w]_i,
\end{align*}
where $A$ and $\tilde{A}$ are the same as defined in the linear-in-means model example from \eqref{eq: lim}. 
%This should be familiar to us as the same feature appears in the linear-in-means model example. 
Note that the summand is 1 if and only if $A_{ij}$, $A_{jk}$ and $w_k$ are all 1s. In other words, if we ignore the normalizing terms, the sum essentially represents the number of length-2 paths in $G$ that start at unit $i$ and arrive at a treated unit. With the normalizing terms, it is close to the fraction of such paths among all length-2 paths that start at unit $i$. Clearly, this feature is informative for unit $i$'s 2-hop neighborhood. 
\end{example}
The above example shows the power of recursion. It allows us to have access to information about much larger neighborhoods without actually looking up all units in larger neighborhoods. In fact, the ReFeX component of ReFeX-LASSO is very efficient in terms of computational complexity \citep{refex}, making the procedure ideal for large-scale experiments on online platforms where network interference is ubiquitous. Another advantage of our algorithm is that all the covariates generated are model-agnostic or model-free---we do not generate them according to any particular response model (or graph model). Since the aggregation functions are arbitrary, ReFeX can quickly generate a very large number of features, even for modest iterations budgets $T$. Despite the fraction of treated neighbors we just saw, we are also able to get the number of treated neighbors for each unit by using sum as the aggregation function. In general, using more complicated aggregation functions yields more complicated features. Thus, the recursive step offers rich features for each unit.  

With minor modifications we can see that all pruning steps in our procedure can be grouped together and done \textit{ex ante}, i.e., before running the experiment and observing the outcomes. Then, after the experiment, we use the observed outcomes to select covariates among all the covariates we have generated. This method has certain advantages, so for completeness we give such a modified version of ReFeX-LASSO below in Algorithm~\ref{alg: post-ReFeX-lasso}, calling it post-ReFeX-LASSO.

\begin{algorithm}
\caption{post-ReFeX-LASSO} \label{alg: post-ReFeX-lasso}
\begin{algorithmic}[1]
\Require Graph $G = (V, \mathcal{E})$, assignment vector $w \in \{0, 1\}^n$, maximum number of iterations $T$.
\Ensure A set of covariates $S$.
\State Initialize $S = \{\}$.
\State For each node/unit $i$, construct $m$ base features and add $m$ base features to $S$.
\For{$t = 1$ to $T$}
\State Recursively construct features by performing aggregations of features in $S$ that were added in the previous iteration over neighbors in $1$-hop neighborhood.
\State Add those newly constructed features to $S$.
\EndFor
\State Regress $y$ on $w$ as well as features from $S$ using LASSO.
\State Keep selected features in $S$ and remove other features from $S$.
\State Return $S$.
\end{algorithmic}
\end{algorithm}

An operational advantage of post-ReFeX-LASSO is that two parts of the algorithm, feature generation and selection, can be done separately. However, in practice we find that post-ReFeX-LASSO leads to estimates with larger variance. Our explanation for this increased variance is two-fold. First, since the number of features generated from ReFeX may be large, separating the generation step and the selection step seems to make the selection step unstable. Second, many of the features generated along the way of post-ReFeX-LASSO are correlated and including all of them simultaneously leads to greater uncertainty in terms of features being selected. Hence, it leads to estimates with larger variance and we recommend ReFeX-LASSO over post-ReFeX-LASSO in all use cases when operationally feasible.

\section{Inference with model-free covariates}

In the previous section, we gave a sequential procedure that outputs a set of covariates $S$ that can be used for regression adjustments when estimating GATEs. This section devotes to inference with model-free covariates. We first discuss how to use model-free covariates returned from ReFeX-LASSO or post-ReFeX-LASSO to do regression adjustment. Following that, we show one selection property of ReFeX-LASSO. We then give theoretical properties of regression adjustment estimator of the GATE using model-free covariates as well as a simple way to construct confidence interval for $\tau$. 

\subsection{Estimation} 
Let $u_i^1, \cdots, u_i^K$ denote the $K$ covariates returned by ReFeX-LASSO or post-ReFeX-LASSO for unit $i$ and let $u_i = \left[u_i^1, \cdots, u_i^K\right]^T \in \mathbb{R}^K$ be the whole feature vector for unit $i$. We further let $\hat{g}$ be the function that maps $(i, w, x_i, G)$ to $u_i$ for each unit $i$. Finally, we denote by $n_c$ the number of control units and $n_t$ the number of treated units with $n_c + n_t = n$.

To estimate the GATE, we fit two linear models on control and treated units using $u_i$'s. Ideally, we hope that there exist vectors $\beta_0, \beta_1$ such that $\beta_0^Tu_i$ and $\beta_1^Tu_i$ are good approximations of $f_0$ and $f_1$. To be specific, we first run an ordinary least squares with observations that are from the control group only and obtain $\hat{\beta}_0$. We then run ordinary least squares again, but now with observations that are from treatment group only and obtain $\hat{\beta}_1$. Meanwhile, the features $u_i$ are all features under the treatment assignment $w$ for which the responses were collected. To estimate the GATE, we are interested not in the response under $u_i$ as it was, but $u_i$ as it would be if $w=\mathbf{0}$ or $w=\mathbf{1}$. We thus pass $\mathbf{0}$ and $\mathbf{1}$ to $\hat{g}$ to obtain the feature vectors $u_i^{gc}$ and $u_i^{gt}$ under global control and global treatment, respectively. 

Combing the coefficient estimates $\hat{\beta}_1$ and $\hat{\beta}_0$ with the vectors $u_i^{gc}$ and $u_i^{gt}$, our estimate of the GATE is then simply
\begin{equation}
    \hat{\tau} = \frac{1}{n}\sum_{i = 1}^{n}(\hat{\beta}_1^Tu_i^{gt} - \hat{\beta}_0^Tu_i^{gc}). \label{eq: reg_adj_est}
\end{equation}
Though assuming a linear model is restrictive, as we discussed previously, if we are able
to generate predictive features then the linear model can be a good approximation to the
true model. ReFeX-LASSO or post-ReFeX-LASSO helps us choose good features to adjust for and thus both reduce the variance of the estimate\footnote{In fact, in the case of no interference, \cite{lin2013} shows that doing linear adjustment can only improve the precision.} %compared to the estimate from cluster randomization 
and reduce the bias we typically incur when ignoring interference.

\subsection{Selection properties} 
Before we delve into inference details, we first discuss selection properties of ReFeX-LASSO, drawing inspiration from prior work on Sequential LASSO~\citep{sequentialLasso}. To this end, we introduce some additional notation. For each iteration $t$, let $\{u_1^t, u_2^t, \cdots, u_{i_t}^t\}$ be the set of features generated in the ReFeX step of ReFeX-LASSO and $s_{*t}$ be the selected features at the $t$-th iteration (note that $s_{*t}$ may contain features that were selected in previous iterations and thus are not in the set $\{u_1^t, u_2^t, \cdots, u_{i_t}^t\}$). Moreover, we let $\mathcal{R}(s)$ to denote the space spanned by features in $s$.

\begin{proposition} \label{prop: selection1}
For $t \geq 1$ and any $j \in \{1, \cdots, i_{t+1}\}$, if $u_j^{t+1} \in \mathcal{R}(s_{*t})$ then $j \notin s_{*(t+1)}$.
\end{proposition}

This first proposition implies two things. First, we have a full rank design matrix at each iteration. Second, the subsequent selection will disregard the features that are highly correlated with the existing ones and hence provides intuition for why the post-ReFeX-LASSO leads to estimate with high variance. Without the sequential procedure of (non-post-) ReFeX-LASSO, two highly correlated features may enter the selection stage together. 

\begin{proposition} \label{prop: selection2}
Our selection is nested in the sense that $s_{*1} \subseteq s_{*2} \subseteq \cdots \subseteq s_{*T}$.
\end{proposition}

This second proposition is relatively self-explanatory and ensures that the sequential procedure actually provides nested feature sets, i.e., by excluding penalties on selected features, we are able to keep them in our feature set $S$. Though our selection procedure in ReFeX-LASSO is quite different from Sequential LASSO \citep{sequentialLasso}, the proofs of the above two propositions are analogous to those in \cite{sequentialLasso}. There are two key differences between our selection procedure and Sequential LASSO. First, instead of keeping all the features for every iteration, we throw away non-selected features in previous iterations. Second, the features under consideration at each iteration are newly generated features rather than existing features. Put another way, we find that the analysis in \cite{sequentialLasso} is robust to such a change in procedure. Note that Sequential LASSO can be used for post-ReFeX-LASSO (but not ReFeX-LASSO) since for post-ReFeX-LASSO we generate all the candidate features in advance.
These two propositions together establish two intuitive properties of our selection step in ReFeX-LASSO that we should expect to hold for our purpose. Their proofs can be found in Appendix~\ref{appendix: proofs}.

\subsection{Consistency} 
We now prove that post-ReFeX-LASSO leads to a consistent estimator of the GATE under standard assumptions one would require for consistency of LASSO. For each unit $i$, we denote the set of features generated by the ReFeX step in post-ReFeX-LASSO as $\{u_i^1, \cdots, u_i^M\}$. We drop the subscript $i$ when we refer to the $j$th feature vector, i.e., $u^j = [u_1^j, \cdots, u_n^j]^T$. Furthermore, we assume that there exists a subset $S_{*} \subset \{u^1, \cdots, u^M\}$ with $|S_{*}| = s$ such that both $f_0$ and $f_1$ are linear in features in $S_{*}$ with coefficient vectors $\beta_0$ and $\beta_1$ respectively. Finally, we denote the design matrix when estimating $\beta_0$ by $U^0$ and the design matrix when estimating $\beta_1$ by $U^1$.
\begin{theorem} \label{thm: consistency}
Suppose that there exists a constant $C > 0$ such that
\begin{equation*}
\max_{j = 1, \cdots, M}\frac{\|u^j\|_2}{\sqrt{n}} \leq C,
\end{equation*}
and the two design matrices $U^0$ and $U^1$ satisfy the $(\kappa; 3)$-RE condition over $S$, then $\hat{\tau}$ is consistent for $\tau$.
\end{theorem}
A proof of Theorem~\ref{thm: consistency} appears in Appendix~\ref{appendix: proofs} and uses mostly standard tools for the study of LASSO $\ell_2$-error bounds \citep{wainwright_2019}. The restricted eigenvalue (RE) condition in Theorem~\ref{thm: consistency} is a standard assumption when proving $\ell_2$-error bound on the coefficient vector. It restricts the curvature for a specific subset of vectors in the Euclidean space. It is defined as follows \citep{Bickel2009, vanderGeer09, JMLR:v11:raskutti10a}
\begin{definition}
The matrix $\mathbf{X}$ satisfies the restricted eigenvalue (RE) condition over $S$ with parameters $(\kappa; \alpha)$ if
\begin{equation*}
    \frac{1}{n}\|\mathbf{X}\Delta\|_2^2 \geq \kappa\|\Delta\|_2^2 \qquad \text{for all } \Delta \in \mathbb{C}_\alpha(S),
\end{equation*}
where $\mathbb{C}_\alpha(S) \coloneqq \{\Delta \in \mathbb{R}^d \, | \, \|\Delta_{S^c}\|_1 \leq \alpha \|\Delta_S\|_1\}$.
\end{definition}

Under the assumptions of Theorem~\ref{thm: consistency}, we are now able to prove GATE consistency under LASSO-based feature selection in at least simple settings such as the following, an example setting where our feature generation procedure outputs two simple features.

\begin{proposition} \label{prop: consistency-example}
Suppose we run a Bernoulli randomized experiment with treatment probability $0 < p < 1$ and we only generate two features, the fraction of treated neighbors $\rho_i$ and number of treated neighbors $\nu_i$. Furthermore, suppose the graph $G$ consists of disjoint cliques of size $3 \leq m_c \leq M$ ($m_c$ is the size of the $c$-th cluster) for some positive constant $M \geq 3$. If the true $f_0$ and $f_1$ are only linear in $\rho_i$, then $\hat{\tau}$ is consistent for $\tau$.
\end{proposition}

The lower bound on $m_c$ is for identifiability since when all clusters have size 2 then $\rho_i$ and $\nu_i$ are essentially the same and we end up with completely duplicated features. Notice also that when all $m_c$'s are equal, we end up with perfect co-linearity so in that case we wouldn't consider distinguishing between these two features. While the above result applies only in a simple setting, it is of its own importance. In practice, it is not uncommon to adjust for fraction of treated neighbors and report the resulting estimate as the estimate of the GATE \citep{ego2019, karrer2021}. The above proposition shows that when we only want to distinguish covariates between fraction of treated neighbors and number of treated neighbors, LASSO is a handy tool.
%Here we point out that the assumption on features is very mild. The assumption here is usually referred to $C-$column normalization assumption. As long as the neighborhood size for every unit and the maximum number of iterations $T$ are either bounded or do not grow fast enough, the $C-$column normalization assumption is almost always satisfied. For fast rate consistency, the assumption we need is slightly stronger. In particular, we have the following theorem:
%\begin{theorem} \label{thm: fast-rate-consistency}
%Suppose that there exists $C_1 > 0$ such that $\|u_i\|_2 \leq C_1$, the network has bounded degree and $\text{Var}(X_i^j) \geq c > 0$ for $i = 1, %\cdots, n$ and $j = 1, \cdots, p$, then we can show that the LASSO estimator for GATE is consistent.
%\end{theorem}
%The assumptions here seem to be much stronger than that in the previous theorem. However, the variance assumption is almost always satisfied besides some pedagogical cases as long as $\exists \mu > 0$ such that $\mu \leq \mathbb{P}(W_i = 1) \leq 1 - \mu$ for each unit $i$. In practice, this requirement on the design is also almost always satisfied especially in the scenario of A/B testing on social networks where network interference is concerned most.

\subsection{Confidence interval via a block bootstrap} \label{subsec: block-bootstrap}
Researchers are usually not just interested in a point estimate of the GATE, they also want to know the uncertainty contained in the estimate, e.g., through confidence intervals. ReFeX-LASSO brings flexibility in doing regression adjustment for GATE estimation, but there is no free lunch and it also brings us
difficulty in doing inference, i.e., in constructing confidence interval for $\tau$. First, unlike \cite{Chin+2019} where one assumes an oracle model, here the true model is unknown. Second, features constructed in \cite{Chin+2019} do not use the observed outcomes. With ReFeX-LASSO, though all the features constructed from ReFeX do not use the outcomes, our selections of covariates \textit{depend} on the realized outcomes.
%can derive clear expression of the variance of the oracle regression adjustment estimator \jucom{Do any of the Chin 2019 features use outcomes? No, right? That is hard even with oracle model, right?}, 
Therefore, ReFeX-LASSO leads to an estimator with no clear variance expression. Moreover, since our final estimate depends on the actual selected covariates, we require some technique analogous to post-selection inference as in \cite{10.1214/15-AOS1371}. \cite{10.1214/15-AOS1371} consider confidence intervals of coefficients conditional on being selected by LASSO. Yet we are interested in the confidence interval of $\tau$, not the coefficients, where our estimate $\hat \tau$ is calculated based on the estimated coefficients as well as selected covariates. Because of the combination of these complexities, we are not able to simply import any known results for inference in this setting.
%our inference problem is unique in the sense that there is no handy tool in the literature that provides \textit{a} solution.

Let us consider the nature of the inference problem we are facing. In general, the randomness of our estimate is incurred not just by the randomness of the potential outcomes but also by the randomness of the assignment vector. To construct the confidence interval, we need to quantify how these two resources of randomness affect our estimate of the GATE. Note that since we know the distribution of the assignment vector, the distribution of a given feature is in fact known. What we don't have a good characterization of is the randomness of the selection procedure incurred by the randomness of the assignment vector. In other words, we require understanding how the random assignments affect the feature selection procedure.

To tackle this complication, we introduce a way to construct confidence intervals based on a block bootstrap. Ideally if we can do the experiment infinitely many times, we could run $2^n$ experiments and calculate $2^n$ estimates of the GATE. A confidence interval for $\tau$ could then be derived easily. Our obvious difficulty is then how should we use one single sample to approximate the sample randomness. We turn to the block bootstrap \citep{bootstrap, efron1994introduction, cameron2008}. The intuition of this usage is that features of units are correlated according to the particular graph structure of $G$ and hence by sampling clusters (which we expect to be relatively disconnected) we are able to keep the bootstrap sample looking like the original sample. On the other hand, resampling units will fail as it cannot replicate the underlying correlation structure in the data. Though we do not provide theoretical guarantees, we will show that in practice the coverage is good and the resulting confidence intervals are of reasonable width. We also note in passing that recent results in \cite{networkbootstrap} demonstrate that there is a version of block bootstrap that does provide theoretical guarantee for certain highlu stylized network processes.

%\jucom{I think we should motivate why the non-block bootstrap is obviously bad, and why we're shooting for blocks here. There's a sentence on the next page, ``The intuition is that features of units are correlated ...'' that sorta starts saying this, but feel like it'd help to explain up front.}

\begin{example} \label{ex: blockbootstrap}
Consider the case where our social network $G$ consists of $C$ disjoint cliques $\mathcal{C}_1, \cdots, \mathcal{C}_C$ of size $m$. Units are fully connected within each clique. This setup can be viewed as a special case of the household experiment studied in \cite{basse2018}. In such a case it is natural to consider sampling all $C$ cliques with replacement to get a bootstrap sample. For network dependent processes satisfying certain technical assumptions, this sampling process is the correct thing to do using arguments in \cite{networkbootstrap}. Suppose we have a network dependent process $\{Y_n, G_n\}$ that satisfies assumptions in \cite{networkbootstrap}. To make block bootstrap consistent, i.e., producing a confidence interval that is consistent in level, Assumption 4.1 in \cite{networkbootstrap} needs to hold.
Tersely employing the notation of that assumption, it is easy to verify that in our case, $\delta_n(s_n) = m$, $\Delta_n(s_n, 2) = 0$, and $D_n(s_n) = m$ for $\forall s_n \geq \max_{c}\text{diam}(\mathcal{C}_c)$, since our graph consists of non-overlapping blocks with equal size $m$. Moreover,
\begin{equation*}
    \omega_n(i, j) = \begin{cases}
    1 \quad \text{if $i$ and $j$ are in the same cluster}, \\
    0 \quad \text{otherwise}.
    \end{cases}
\end{equation*}
and $\omega_n(j) = 1$ for all $j \in [n]$. With these values, we immediately see that the Assumption 4.1 in \cite{networkbootstrap} holds as long as $m = o(n)$. Since the only remaining assumptions needed to make block bootstrap consistent are about the network dependent process itself, we can conclude that block bootstrap would be valid in this toy model for network dependent processes given in \cite{networkbootstrap}.
\end{example}

We present two versions of block bootstrap here, one for regression adjustment with post-ReFeX-LASSO and one for regression adjustment with ReFeX-LASSO. Before actually giving the two block bootstrap procedures, we first introduce the key ingredient in our block bootstrap procedure, a randomized graph clustering algorithm. Our block bootstrap procedure involves partition the graph into several clusters. The generic algorithm we use is $k$-hop-max clustering \citep{uganderyin2020}, a simple adaptation of the CKR partitioning algorithm \citep{calinescu2005approximation}. The details are shown in Algorithm~\ref{Alg:1_hop_max}. The algorithm provides a random clustering of the graph that depends on random initial conditions. The algorithm is light in computation when $k = 1$ as we only need to look at one's direct neighbors. Also, it returns neighborhood-like clusters. As a remark connecting back to above example, if our graph consists of disjoint fully connected clusters then 1-hop max clustering is able to return exactly these clusters as final output. In general, when $k > 1$, we obtain larger clusters that are centered around fewer nodes. 

\begin{algorithm}[ht]
\caption{$k$-hop-max graph clustering}
\begin{algorithmic}[1]
\Require Graph $G = (V, E)$.
    % \KwIn{Unweighted graph $G = (V, E)$, integer $r \geq 1$}
    \Ensure Graph clustering $\mathcal{C}_1, \cdots, \mathcal{C}_c$.
    \For{$i \in V$}
     \State $X_i \leftarrow \mathcal U(0, 1)$;
    \EndFor
    \For{$i \in V$}
      % $c_i \gets \max([X_j \text{ for } j \in B_r(i)])$\;
     \State $i \leftarrow \text{argmax}([X_j \text{ for } j \in B_k(i)])$;
    \EndFor
    \State Return $\mathcal{C}_1, \cdots, \mathcal{C}_c$.
\end{algorithmic}
\label{Alg:1_hop_max}
\end{algorithm}

We first present the block bootstrap procedure for post-ReFeX-LASSO, given in Algorithm~\ref{alg: post-ReFeX-lasso-boot}. With post-ReFeX-LASSO, the bootstrap procedure is simpler since the feature generation and selection part are separated. Unlike the usual bootstrap where we sample random individual units with replacement, here we sample random clusters from the graph clustering algorithm with replacement. The intuition is that features $u_i$ of units are correlated according to the particular graph structure of $G$ and hence by sampling clusters, which we expect to be relatively disconnected, we are able to keep the bootstrap sample ``looking like'' the original sample. As a specific caveat, though in expectation the bootstrap sample has sample size $n$, if we do not have uniformly sized clusters, then the bootstrap sample may end up with much larger or smaller sample size. Hence we run the graph clustering algorithm $l$ times and for each clustering we run block bootstrap with the number of bootstrap replicates $B$. We use $k = T^* + 1$ for $k$-hop-max clustering in Algorithm~\ref{alg: post-ReFeX-lasso-boot} where $T^*$ is the number of iteration where there were features still got selected (since if no feature got selected in the $(T^* + 1)$-th iteration then interference should happen within $(T^* + 1)$-hop neighborhood).

\begin{algorithm}
\caption{Block bootstrap for post-ReFeX-LASSO} \label{alg: post-ReFeX-lasso-boot}
\begin{algorithmic}[1]
\Require Graph $G = (V, \mathcal{E})$, assignment vector $w \in \{0, 1\}^n$, number of bootstrap samples $B$.
\Ensure Confidence interval for $\tau$.
\State Collect the assignment $w_i$, features $u_i^1, \cdots, u_i^M$ in $S$ generated before running LASSO, outcome $y_i$ for each unit $i$. Record the maximum iteration number $T^*$ where one of the features generated at that iteration was selected.
\State Use $k$-hop max clustering algorithm with $k = T^* + 1$ to divide $n$ units into $C$ clusters $\mathcal{C}_1, \cdots, \mathcal{C}_C$.
\For{$b = 1$ to $B$}
\State Sample $C$ clusters with replacement from $\mathcal{C}_1, \cdots, \mathcal{C}_C$.
\State Construct the $b$-th bootstrap sample $(w^b, u^{1, b}, \cdots, u^{M, b}, y^b)$ with units from sampled clusters.
\State Regress $y$ on $w$ as well as $M$ features using LASSO.
\State Compute the estimate $\hat{\tau}^b$ using selected features and the bootstrap sample.
\EndFor 
\State Repeat line 2-8 for $\ell$ times and obtain $\ell\cdot B$ bootstrap estimates in total.
\State Compute the $\alpha/2$-th quantile $q_{\alpha/2}^{*}$ and the $(1 - \alpha/2)$-th quantile $q_{1 - \alpha/2}^{*}$ of the sample of all bootstrap estimates $\hat{\tau}^1, \cdots, \hat{\tau}^{\ell B}$.
\State Return $\left[q_{\alpha/2}^{*}, q_{1 - \alpha/2}^{*}\right]$ as the $(1-\alpha) \times 100\%$ confidence interval for $\tau$.
\end{algorithmic}
\end{algorithm}

Next we present the version of block bootstrap with ReFeX-LASSO, given in Algorithm~\ref{alg: ReFeX-lasso-boot}. Note that we cannot simply use the same algorithm since it performs feature generation and feature selection concurrently. Compared to Algorithm~\ref{alg: post-ReFeX-lasso-boot}, $T^*$ now represents the stopping time of ReFeX-LASSO. Meanwhile, similar to Algorithm~\ref{alg: post-ReFeX-lasso-boot}, the bootstrap sample is only used in the feature selection step of ReFeX-LASSO. That being said, for each iteration, we still use the same graph $G$ to generate features but then we use the bootstrap sample of these features to do selection. The intuition behind using the original graph is that we view the graph as fixed and the correlation structure of all features are then induced by this graph. Therefore, we do not paste all sampled clusters together to form a new graph to generate features for next iteration. On the other hand, if we do believe that the graph is generated from some random process then we may also reconstruct the graph from sampled units by pasting all sampled clusters together.

\begin{algorithm}[ht]
\caption{Block bootstrap for ReFeX-LASSO} \label{alg: ReFeX-lasso-boot}
\begin{algorithmic}[1]
\Require Graph $G = (V, \mathcal{E})$, assignment vector $w \in \{0, 1\}^n$, number of bootstrap samples $B$.
\Ensure Confidence interval for $\tau$.
\State Collect the assignment $w_i$ and %base features $u_i^1, \cdots, u_i^m$ used in ReFeX-LASSO, 
outcome $y_i$ for each unit $i$. Record the stopping time for ReFeX-LASSO $T^*$.
\State Use $k$-hop max clustering with $k = T^* + 1$ to divide $n$ units into $C$ clusters $\mathcal{C}_1, \cdots, \mathcal{C}_C$.
\For{$b = 1$ to $B$}
\State Sample $C$ clusters with replacement from $\mathcal{C}_1, \cdots, \mathcal{C}_C$.
\State Construct the $b$-th bootstrap sample with units from sampled clusters.
\State Rerun ReFeX-LASSO with the original sample for feature generation and the bootstrap sample for feature selection. 
\State Use the covariates returned from last step as well as the bootstrap sample to get estimate of $\tau$, $\hat{\tau}^b$.
\EndFor 
\State Repeat line 2-8 for $\ell$ times and obtain $\ell\cdot B$ bootstrap estimates in total.
\State Compute the $\alpha/2$-th quantile $q_{\alpha/2}^{*}$ and the $(1 - \alpha/2)$-th quantile $q_{1 - \alpha/2}^{*}$ of the sample of all bootstrap estimates $\hat{\tau}^1, \cdots, \hat{\tau}^{\ell B}$.
\State Return $\left[q_{\alpha/2}^{*}, q_{1 - \alpha/2}^{*}\right]$ as the $(1-\alpha) \times 100\%$ confidence interval for $\tau$.
\end{algorithmic}
\end{algorithm}

In the above two algorithms, we utilize a randomized graph clustering algorithm that can be easily implemented. Of course, this is not the only possible choice for the graph clustering algorithm one can use. We note by passing that there are many graph clustering algorithms available for practitioners \citep{nishimuraugander13, spielman13, ameljohan20, shi2020} that exhibit various properties.

We conclude this section with a discussion of how to suitably choose the sizes of clusters. We consider three scenarios and show why they may fail with heuristics from \cite{networkbootstrap}. Though we are not considering the same problem as in \cite{networkbootstrap}, given that we have a more complicated setup, we do not expect that weaker assumptions than those in \cite{networkbootstrap} would be sufficient for good coverage in our case. Therefore, we view assumptions in \cite{networkbootstrap} as what we should expect to have in order to make our block bootstrap consistent. 

The first scenario that we consider is when we have $O(n)$ clusters with non-constant sizes. Then the second absolute central moment of block sizes may be non-vanishing as $n \rightarrow \infty$ but the average block size is $O(1)$. This implies that unless the clusters are relatively uniform, there would be a violation to Assumption 4.1 in \cite{networkbootstrap}. As a second scenario, consider the case when we have $O(1)$ clusters. Now the maximum block size must be of order $O(n)$ and the average block size is at most $O(n)$, hence Assumption 4.1 in \cite{networkbootstrap} is certainly violated. In general, we don't want to have too many clusters or too few clusters. Finally, then, consider a scenario where we have $\sqrt{n} - 1$ clusters of size $\sqrt{n}$ and $\sqrt{n}$ clusters of size 1. Now the average block size is of order $O(n^{1/2})$ and the second absolute central moment of block sizes is not of lower order, which implies that the ratio does not vanish as $n \rightarrow 0$ and again Assumption 4.1 in \cite{networkbootstrap} is violated. This last example shows that the cluster sizes are not simply a matter of avoiding too big/small or few/many clusters, but instead here we see we cannot have two groups of clusters with different size magnitudes. In summary, the advice is to use a reasonable number of clusters that have sizes of roughly the same magnitude. What we present in Algorithm~\ref{alg: post-ReFeX-lasso-boot} and \ref{alg: ReFeX-lasso-boot} are good default choices if the network is not very dense.

\section{Simulation experiments}
In this section, we use simulations to provide both empirical guidance on our method when theory is lacking and empirical evidence of the usefulness of our method. We make use of the Facebook 100 dataset \citep{TRAUD20124165} of real-world social networks. The networks in this dataset are complete online friendship networks for one hundred colleges and universities collected from a single-day snapshot of Facebook in September 2005. For our simulations we use the network of Swarthmore college students, being of modest size. We extract the largest connected components of the Swarthmore network, obtaining a social network with 1,657 nodes and 61,049 edges. The diameter of the network is 6 and the average pairwise distance is 2.32. Since this network is quite dense, estimation of the GATE would be very difficult when interference is strong. We use this network to demonstrate that even for such a network, we are still able to get relatively good estimates from (post-) ReFeX-LASSO.

We generate an assignment vector using a Bernoulli design with success probability $0.5$ and generate outcome variables according to certain models with varying magnitude of network interference; these models are summarized in Table \ref{table:outcome-model-linear} and \ref{table:outcome-model-truncated}. We will discuss in detail about these outcome models in Section~\ref{subsec:outcome-models}. Our simulations can be viewed as semi-synthetic experiments---we use a true social network but we generate outcomes according to specified models.

Section~\ref{subsec: sim-estimation-gate} introduces the baseline estimators that we compare with in our simulations. Section~\ref{subsec:outcome-models} discusses the outcome models that we use for generating the outcomes with various degree of interference. Section~\ref{subsec: estimation-sim} compares the regression adjustment estimator using model-free covariates with those commonly-used estimators in practice as in Section~\ref{subsec: sim-estimation-gate} and demonstrate that it has good performance in terms of root mean squared error. Section~\ref{subsec: sim-ci-gate} explores the empirical performance of the confidence interval constructed via block bootstrap and discusses some practical aspects in the procedure.

\subsection{Estimation of the GATE} \label{subsec: sim-estimation-gate}
Our ultimate goal of constructing model-free covariates is to use them in GATE estimation. We first explore the empirical performance of the regression adjustment estimator using model-free covariates. Specifically, we compare it with two kinds of estimators that are commonly used in practice: (i) the difference-in-mean estimator and (ii) a H\'{a}jek estimator under a network exposure model \citep{manski2013}. Difference-in-mean estimator calculate the difference between average outcome among treated units and average outcome among control units:
\begin{equation*}
    \hat{\tau}^{DM} = \frac{1}{\sum_{i=1}^nW_i}\sum_{i = 1}^n Y_iW_i - \frac{1}{\sum_{i=1}^n(1-W_i)}\sum_{i = 1}^n Y_i(1-W_i). 
\end{equation*}
Obviously this estimator ignores interference and will thus incur large bias when interference is significant. 

The basic H\'{a}jek estimator for the ATE is defined as
\begin{equation*}
    \hat{\tau}^{\text{H\'{a}jek}} = \frac{\sum_{i = 1}^nY_iW_i/\mathbb{P}(W_i = 1)}{\sum_{i = 1}^n\mathbb{I}(W_i=1)/\mathbb{P}(W_i = 1)} - \frac{\sum_{i = 1}^nY_i(1-W_i)/\mathbb{P}(W_i = 0)}{\sum_{i = 1}^n\mathbb{I}(W_i=0)/\mathbb{P}(W_i = 0)}.
\end{equation*}
Here we will consider a version of H\'{a}jek estimator that accounts for interference. \cite{manski2013} studies identification of potential outcome distributions under interference. One concrete example is when one's outcome only depends on one's own assignment as well as the distribution of assignments for his/her neighbors. \cite{UganderKBK13} further considers a fractional exposure model where it is assumed that if one is treated and a $q > 0.5$ fraction of one's neighbors are treated then one's outcome is equal to the potential outcome associated with the assignment vector $\mathbf{1}$. Similarly, in this exposure model if one is not treated and one's fraction of treated neighbors is at most $1-q$ then one's outcome is equal to the potential outcome associated with the assignment vector $\mathbf{0}$. Formally, $\forall w, w' \in \{0, 1\}^n$, this fractional exposure model assumes: 
\begin{equation*}
  w_i = 1, \frac{1}{|\mathcal{N}_i|}\sum_{j \in \mathcal{N}_i}w_j \geq q \implies Y_i(w) = Y_i(\mathbf{1}),
\end{equation*}
and
\begin{equation*}
  w_i = 0, \frac{1}{|\mathcal{N}_i|}\sum_{j \in \mathcal{N}_i}w_j \leq 1 - q \implies Y_i(w) = Y_i(\mathbf{0}).
\end{equation*}
We can then use a H\'{a}jek estimator that corrects for the probability that these conditions are met under a Bernoulli design. Specifically, we define the events $E_i^{1, q}=\{W_i = 1, \frac{1}{|\mathcal{N}_i|}\sum_{j \in \mathcal{N}_i}w_j \geq q\}$ and  $E_i^{0, 1-q}=\{W_i = 0, \frac{1}{|\mathcal{N}_i|}\sum_{j \in \mathcal{N}_i}w_j \leq 1-q\}$. The corresponding H\'{a}jek estimator under a fractional exposure model is then
\begin{equation} 
    \hat{\tau}^{\text{H\'{a}jek}}_{q, 1-q} = \frac{\sum_{i = 1}^nY_i\mathbb{I}(E_i^{1, q})/\mathbb{P}(E_i^{1, q})}{\sum_{i = 1}^n\mathbb{I}(E_i^{1, q})/\mathbb{P}(E_i^{1, q})} - \frac{\sum_{i = 1}^nY_i\mathbb{I}(E_i^{0, 1-q})/\mathbb{P}(E_i^{0, 1-q})}{\sum_{i = 1}^n\mathbb{I}(E_i^{0, 1-q})/\mathbb{P}(E_i^{0, 1-q})}. \label{sim-est: hajek}
\end{equation}

This estimator accounts for interference by taking the assignments of direct neighbors into consideration. If we still assume local interference in the sense that only one's direct neighbors can impact one's response but want a fully agnostic setting then we could choose $q = 1$ (notice that in this case the H\'{a}jek estimator is consistent). In our case, the number of neighbors one has is usually quite large and under independent Bernoulli assignment we wouldn't expect to observe many units with all neighbors being treated or not treated. As a bias-variance compromise, we choose $q = 0.8$. 

Finally, we also compare our (post-) ReFeX-LASSO regression adjustment estimator with two linear regression adjustment estimators that adjust for specific features. We will describe these two estimators in detail later when we present the simulation results in Section~\ref{subsec: estimation-sim}. For post-ReFeX-LASSO and ReFeX-LASSO, we choose $T = 2$ and the base features to be fraction of treated neighbors, number of treated neighbors, fraction of edges in neighborhood that connects a treated unit and a control unit and also fraction of edges in neighborhood that connects a treated unit and a treated unit. For aggregation functions in (post-) ReFeX-LASSO, we use both the mean and variance.

\subsection{Outcome models} \label{subsec:outcome-models}
Here we describing the outcome models we use in our simulation study. We carry forward the notation from as in Proposition~\ref{prop: consistency-example}, using $\rho_i$ to denote the fraction of treated direct neighbors for unit $i$ and $\nu_i$ to denote number of treated direct neighbors. 

We first consider estimation under linear interference. The first model is a linear model in both number of treated neighbors and fraction of treated neighbors. Such model is also considered in \cite{pouget-abadie19} and \cite{Chin+2019}. Specifically,
\begin{equation}
    f_0(w, G) = \alpha_0 + \xi_0\rho_i + \gamma_0\nu_i \label{sim_outcome_model: simple0}
\end{equation}
and
\begin{equation}
    f_1(w, G) = \alpha_1 + \xi_1\rho_i + \gamma_1\nu_i. \label{sim_outcome_model: simple1}
\end{equation}
The difference $\alpha_1 - \alpha_0$ can be viewed as the primary effect of the treatment and coefficients $(\xi_w, \gamma_w)$ for $w = 0, 1$ govern how the unit respond to treatment and control, respectively. In particular, if $\xi_w = \gamma_w = 0$ then there is no interference and we are back to usual setup of ATE estimation under SUTVA. Note that for this model, there is no interference beyond the 1-hop neighborhood and hence the estimation problem is considerably easier. We will refer to this response model as simple linear interference.

Building on the discussion of the linear-in-means model in the introduction, we also consider a response model where the interference propagates out to $k$-hop neighborhoods for $k \geq 2$. This model can be viewed as a truncated linear-in-means model; instead of summing up to infinity, we truncate the model at $j = J$ for some number $J > 1$. 

\begin{table*}[ht]
\centering
\scalebox{1.0} {
\begin{tabular}{c l c c} 
\toprule 
Model type & $(\alpha_0, \alpha_1)$ & $(\xi_0, \xi_1)$ & $(\gamma_0, \gamma_1)$\\
\midrule
    \multirow{1}{*}{Model 0} 
            & (0, 2) & (0, 0) & (0, 0) \\
            \midrule
    \multirow{1}{*}{Model 1} 
            & (0, 2) & (1, 1.5) & (0.005, 0.0025) \\
            \midrule
    \multirow{1}{*}{Model 2} 
            & (0, 2) & (1, 2) & (0.005, 0.01) \\
\bottomrule 
\end{tabular}}
\caption{Parameters of simple linear interference outcome model (\eqref{sim_outcome_model: simple0} and \eqref{sim_outcome_model: simple1}) used in simulation experiments.} 
\label{table:outcome-model-linear} % A label for referencing this table elsewhere, references are used in text as \ref{label}
\end{table*}

\begin{table*}[ht]
\centering
\scalebox{1.0} {
\begin{tabular}{c l c c c} 
\toprule 
Model type & $\alpha$ & $\beta$ & $\gamma$ & $J$\\
\midrule
    \multirow{1}{*}{Model 3} 
            & 1 & 5 & 2 & 2\\
            \midrule
    \multirow{1}{*}{Model 4} 
            & 1 & 5 & 3 & 2\\
            \midrule
    \multirow{1}{*}{Model 5} 
            & 1 & 5 & 1 & 3\\
            \midrule
    \multirow{1}{*}{Model 6} 
            & 1 & 5 & 2 & 3\\
\bottomrule 
\end{tabular}}
\caption{Parameters of truncated linear-in-means outcome model used in simulation experiments.} 
\label{table:outcome-model-truncated} % A label for referencing this table elsewhere, references are used in text as \ref{label}
\end{table*}

Overall we consider the following model configurations of linear interference. Table \ref{table:outcome-model-linear} and \ref{table:outcome-model-truncated} summarize the configurations of the models we consider for simulations. Note that model 0 exhibits no interference. For all models, the error terms are independently normally distributed with variance 1. The true GATE in these outcome models (either by an exact calculation or by a Monte Carlo estimate on the Swarthmore network) are 2, 3.69, 4.74, 15, 20, 15 and 35 respectively.

Beyond linear interference, we also examine a slightly more complicated scenario where linear interference is violated. In particular, we consider $f_0$ and $f_1$ that are nonlinear in $\rho_i$ and $\nu_i$. The nonlinear functions we use are sigmoid-type so that it is hard to approximate by any linear model\footnote{We document this model in the Appendix~\ref{appx: supp}.}. We use the Monte Carlo estimate, 9.55, as the true GATE when reporting the simulation results. Our purpose here is to show that even if we have nonlinear $f_0$ and $f_1$ which violates our linear interference assumption, our method still leads to an estimator with reasonable performance. This also echos our previous discussion. In GATE estimation, we are always predicting for a data point that is outside the range of our observed/training data and hence a simple model can be quite reliable.

\begin{table*}[ht]
\centering
\scalebox{0.9} {
\begin{tabular}{c l c c c c c c c c} 
\toprule 
Estimator & $\hat{\tau}^{\text{DM}}$ & $\hat{\tau}^{\text{H\'ajek}}_{\text{0.8,0.2}}$ & $\hat{\tau}_{\text{frac}}$ & $\hat{\tau}_{\text{num}}$ & post-ReFeX-LASSO & ReFeX-LASSO & $\hat{\tau}_{\text{oracle}}$\\
\midrule
    \multirow{1}{*}{Model 0} 
            &  0.05  & 0.76 & 0.24 & 0.07 & 0.50 & 0.32 & 0.05\\
            \midrule
    \multirow{1}{*}{Model 1} 
            &  1.53  & 1.02 & 0.36 & 1.22 & 1.54 & 0.70 & 0.25\\
            \midrule
    \multirow{1}{*}{Model 2} 
            &  2.06  & 1.41 & 0.47 & 1.49 & 1.49 & 0.59 & 0.24\\
            \midrule
    \multirow{1}{*}{Model 3} 
            &  10.02 & 3.84 & 0.37 & 9.86 & 1.08 & 0.93 & 0.37\\
            \midrule
    \multirow{1}{*}{Model 4} 
            &  15.02  & 5.60 & 0.56 & 14.72 & 1.68 & 1.59 & 0.56\\
            \midrule
    \multirow{1}{*}{Model 5} 
            &  9.92  & 6.61 & 4.53 & 9.82 & 1.38 & 1.73 & 1.98\\
            \midrule
    \multirow{1}{*}{Model 6} 
            & 29.67  & 22.31 & 18.22 & 29.46 & 2.42 & 2.47 & 4.45\\
\bottomrule 
\end{tabular}}
\caption{RMSE of estimators of the GATE assuming linear interference (simple linear interference and truncated linear-in-means) outcome models.} 
\label{table:estimation-results-linear-rmse} % A label for referencing this table elsewhere, references are used in text as \ref{label}
\end{table*}

\begin{table*}[ht]
\centering
\scalebox{0.9} {
\begin{tabular}{c l c c c c c c c c} 
\toprule 
Estimator & $\hat{\tau}^{\text{DM}}$ & $\hat{\tau}^{\text{H\'ajek}}_{\text{0.8,0.2}}$ & $\hat{\tau}_{\text{frac}}$ & $\hat{\tau}_{\text{num}}$ & post-ReFeX-LASSO & ReFeX-LASSO & $\hat{\tau}_{\text{oracle}}$\\
\midrule
    \multirow{1}{*}{Model 0} 
            &  0.004  & 0.120 & 0.021 & 0.009 & 0.106 & 0.042 & 0.004\\
            \midrule
    \multirow{1}{*}{Model 1} 
            &  -1.53  & -0.68 & -0.25 & -1.22 & -0.72 & -0.004 & -0.05\\
            \midrule
    \multirow{1}{*}{Model 2} 
            &  -2.06  & -1.16 & -0.39 & -1.48 & -0.56 & -0.29 & 0.008\\
            \midrule
    \multirow{1}{*}{Model 3} 
            &  -10.02 & -3.67 & -0.01 & -9.85 & 0.28 & 0.19 & -0.01\\
            \midrule
    \multirow{1}{*}{Model 4} 
            &  -15.02  & -5.46 & 0.003 &  -14.72 & 0.36 & 0.31 & 0.03\\
            \midrule
    \multirow{1}{*}{Model 5} 
            &  -9.92  & -6.55 & -4.52 & -9.82 & -0.01 & -0.24 & 0.68\\
            \midrule
    \multirow{1}{*}{Model 6} 
            & -29.67  & -22.28 & -18.21 & -29.45 & -0.21 & -0.31 & 2.77\\
\bottomrule 
\end{tabular}}
\caption{Empirical bias of estimators of the GATE assuming linear interference (simple linear interference and truncated linear-in-means) outcome models.} 
\label{table:estimation-results-linear-bias} % A label for referencing this table elsewhere, references are used in text as \ref{label}
\end{table*}

\subsection{Simulation results} \label{subsec: estimation-sim}
We study both the bias and the root mean squared error (RMSE) of each estimator under these varied models. Table~\ref{table:estimation-results-linear-rmse} and Table~\ref{table:estimation-results-linear-bias} show the RMSE and bias of several different estimators under linear interference. In these two tables, we show results of two kinds of regression adjustment estimators. $\hat{\tau}_{\text{frac}}$ is the regression adjustment estimator that adjusts for the fraction of treated neighbors and $\hat{\tau}_{\text{num}}$ adjusts for the number of treated neighbors. They are also considered in \cite{Chin+2019}. We also show the oracle adjustment estimator $\hat{\tau}_{\text{oracle}}$ as a reference, which marks the best we can do with full knowledge of the response model. Note that in some cases other estimators can perform better than the oracle since the oracle adjustment estimator only means we use oracle control covariates. The covariates are inevitably random and we are not averaging over all possible assignment vectors. Moreover, for the truncated linear-in-means model, the true covariates are highly correlated, causing the oracle adjustment estimator to have a large variance. Finally, $\hat{\tau}^{\text{DM}}$  and $\hat{\tau}^{\text{H\'ajek}}_{\text{0.8,0.2}}$ refer to the simple difference-in-mean estimator and the H\'ajek estimator in Equation~\eqref{sim-est: hajek} with $q = 0.8$ as we mentioned earlier.

First, if we look at the results for Model 0, i.e., when there is no interference, post-ReFeX-LASSO and ReFeX-LASSO all give better performance compared to the H\'ajek estimator. Second, for Model 1 and Model 2, the true interference mechanism is simple linear interference. As we can see from the first two rows of Table~\ref{table:estimation-results-linear-rmse} and Table~\ref{table:estimation-results-linear-bias}, if we fail to account for one feature, the bias and/or the RMSE can be large. Also, ReFeX-LASSO is dominating post-ReFeX-LASSO with significantly lower bias and RMSE since for this case ReFeX-LASSO is able to stop considering further features after the first iteration. For Model 3--6, the underlying model is a truncated linear-in-means model and the only difference between them is the stopping number $J$. For the models with $J = 2$ (Models 3 and 4), the interference is still local, i.e., within one's direct neighbors, but for $J = 3$ (Models 5 and 6), it is crucial to consider information from 2-hop neighbors. Our simulation results verify this intuition. We see that $\hat{\tau}_{\text{frac}}$ is doing well for model 3 and 4 but very poorly for model 5 and 6. Both post-ReFeX-LASSO and ReFeX-LASSO lead to estimators with relatively small bias and small RMSE for these more challenging response models.

%As we can see, the RMSE of ReFeX-LASSO based regression adjustment estimator is comparable to the best we can do - adj3 and when we incorrectly adjust for two correlated features, the variance can be quite large. Moreover, we notice that the bias of our ReFeX-LASSO based regression adjustment estimator is even smaller than adj1 when the influence of number of treated neighbors is non-negligible.

Turning to the nonlinear model, Table~\ref{table:estimation-results-nonlinear} below shows our results there. In this case, $\hat{\tau}_{\text{frac}}$ and $\hat{\tau}_{\text{num}}$ represent the same regression adjustment estimators as in the linear case. Compared to difference-in-means and H\'ajek, ReFeX-LASSO leads to estimator with much better performance. Also, based on the comparison of $\hat{\tau}_{\text{frac}}$, $\hat{\tau}_{\text{num}}$ and ReFeX-LASSO, we see that, as in the linear interference case, even if we happen to adjust for some feature that is of importance, failing to take all relevant features into account will lead to estimators with either large bias, large variance, or both. In other words, ReFeX-LASSO helps one choose which set of features to adjust for and hence incur much smaller bias or variance.

\begin{table*}[ht]
\centering
\scalebox{1.0} {
\begin{tabular}{c l c c c c c c c c} 
\toprule 
Estimator & $\hat{\tau}^{\text{DM}}$  & $\hat{\tau}^{\text{H\'ajek}}_{\text{0.8,0.2}}$ & $\hat{\tau}_{\text{frac}}$ & $\hat{\tau}_{\text{num}}$ & post-ReFeX-LASSO & ReFeX-LASSO\\
\midrule
    \multirow{1}{*}{Bias} 
            &  -5.54  & -2.72 & -1.56 & -2.72 & 1.29 & 1.33\\
            \midrule
    \multirow{1}{*}{RMSE} 
            & 5.55  & 3.73 & 1.92 & 2.73 & 5.68 & 2.75\\
\bottomrule 
\end{tabular}}
\caption{RMSE and empirical bias of estimators of the GATE assuming a nonlinear interference (Appendix~\ref{appx: supp}) outcome model.} 
\label{table:estimation-results-nonlinear} % A label for referencing this table elsewhere, references are used in text as \ref{label}
\end{table*}

From these simulations we take away that ReFeX-LASSO is able to identify influential features for regression adjustment and hence produce an estimator with relatively good performance across many model specifications. We also see that ReFeX-LASSO generally, though not always, performs significantly better than post-ReFeX-LASSO. This is due to the fact that we select features sequentially and hence reduce the variance. In contrast, a standard regression adjustment estimator considered in \cite{Chin+2019} for some network features ($\hat{\tau}_{\text{frac}}$ and $\hat{\tau}_{\text{num}}$ in our simulations) can be far-off if we fail to choose the right feature. Finally, exposure mapping based estimator like the fractional-exposure-H\'ajek estimator can also be pretty bad if we have interference that is quite different from the assumptions of the exposure model that such estimators assume.

\subsection{Confidence interval for the GATE} \label{subsec: sim-ci-gate}
In Section~\ref{subsec: block-bootstrap} we introduced a way to construct a confidence interval for $\tau$ via a block bootstrap and gave an explicit algorithm for graph-based block construction. We now evaluate the empirical coverage of the resulting confidence interval from our block bootstrap. Throughout this section, we focus on 90\% confidence interval for $\tau$. Instead of using the Swarthmore College network as in the previous section, we use the farmer network in \cite{cai2015} where we have a larger and sparser network compared to the Swarthmore College network. In fact, the average size of 2-hop neighborhoods in Swarthmore network is 1092.65 and the average size of 3-hop neighborhoods in Swarthmore network is 1622.27. Hence, if we believe that interference is beyond 1-hop neighborhood, bootstrap will not perform well on such a dense graph since it is hard to create bootstrap samples that respect the structure in the original sample\footnote{We found that the block bootstrap still gives near to nominal coverage on Swarthmore nwtwork when interference is local, i.e., within direct neighbors.}. On the other hand, the farmer network in \cite{cai2015} is less dense with 2-hop neighborhoods having an average size 23.95 and 3-hop neighborhoods having an average size 41.49. We will introduce in more details about the background and the details of this network in Section~\ref{sec:real-data-example}. In general, if the network is too dense to produce well-isolated and balanced clusters then the bootstrap would fail. One thing to notice is that the farmer network itself is associated with a natural clustering based on which village the each farmer lives in, namely, each village can be viewed as a cluster in the network. In our simulations here, we thus also show the results of constructing the confidence interval with block bootstrap of ReFeX-LASSO that uses this ``oracle clustering'' of villages. Finally, since we have a sparser network (making interference easier to manage), we consider two different sets of parameters for linear models that make the effect from number of treated neighbors larger (and thus GATE estimation harder). Table~\ref{table:additional-outcome-model-linear} shows the values of the parameters, loosely based on Model 2 (thus named 2a and 2b)

We first evaluate the effectiveness of such a bootstrap method. We assume linear interference and consider Model 3-6 as well as Model 2a and 2b. We fix $\ell = 3$, $B = 100$ and the coverage is calculated by repeating the whole process 100 times.
\begin{table*}[ht]
\centering
\scalebox{1.0} {
\begin{tabular}{c l c c} 
\toprule 
Model type & $(\alpha_0, \alpha_1)$ & $(\xi_0, \xi_1)$ & $(\gamma_0, \gamma_1)$\\
\midrule
    \multirow{1}{*}{Model 2a} 
            & (0, 2) & (1, 3) & (0.01, 0.025) \\
\midrule
    \multirow{1}{*}{Model 2b} 
            & (0, 2) & (1, 3) & (0.05, 0.15) \\
\bottomrule 
\end{tabular}}
\caption{Additional parameters of simple linear interference model (\eqref{sim_outcome_model: simple0} and \eqref{sim_outcome_model: simple1}) used in simulation experiments.} 
\label{table:additional-outcome-model-linear} % A label for referencing this table elsewhere, references are used in text as \ref{label}
\end{table*}
Table~\ref{table:ci-results-linear-coverage} and \ref{table:ci-results-linear-length} show the coverage and the average length of the confidence intervals constructed from our block bootstrap of post-ReFeX-LASSO and ReFeX-LASSO. To show the necessity of using block bootstrap and of considering the randomness of the assignment vector, we also include the result of constructing confidence interval using a naive bootstrap where we just sample each unit with replacement.

%Under linear interference, we find that the block bootstrap leads to a confidence interval with correct coverage and of reasonable length. 

\begin{table*}[ht]
\centering
\scalebox{0.85} {
\begin{tabular}{c c c c c} 
\toprule 
Model & post-ReFeX-LASSO & ReFeX-LASSO & Naive Bootstrap & Bootstrap with oracle clustering\\
\midrule
    %\multirow{1}{*}{Model 1} 
    %        &   & 88\%\\
    %        \midrule
    \multirow{1}{*}{Model 2a} 
            & 93\% & 92\% & 94\% & 92\%\\
            \midrule
    \multirow{1}{*}{Model 2b} 
            & 92\% & 96\% & 93\% & 95\%\\
            \midrule
    \multirow{1}{*}{Model 3} 
            & 90\%  & 90\% & 83\% & 91\%\\
            \midrule
    \multirow{1}{*}{Model 4} 
            & 88\%  & 87\% & 80\% & 91\%\\
            \midrule
    \multirow{1}{*}{Model 5} 
            & 91\% & 93\% & 84\% & 92\% \\
            \midrule
    \multirow{1}{*}{Model 6} 
            & 90\% & 91\% & 67\% & 93\%\\
\bottomrule 
\end{tabular}}
\caption{Coverage of different bootstrap 90\% confidence intervals for the GATE with linear interference (simple linear interference and truncated linear-in-means) outcome models.} 
\label{table:ci-results-linear-coverage} % A label for referencing this table elsewhere, references are used in text as \ref{label}
\end{table*}

\begin{table*}[ht]
\centering
\scalebox{0.85} {
\begin{tabular}{c c c c c} 
\toprule 
Model & post-ReFeX-LASSO & ReFeX-LASSO & Naive Bootstrap & Bootstrap with oracle clustering\\
\midrule
    \multirow{1}{*}{Model 2a} 
            &0.245 & 0.220 & 0.235 & 0.228\\
            \midrule
    \multirow{1}{*}{Model 2b} 
            & 0.435 & 0.384  & 0.390 & 0.382\\
            \midrule
    \multirow{1}{*}{Model 3} 
            & 0.403 & 0.380 & 0.330 & 0.414\\
            \midrule
    \multirow{1}{*}{Model 4} 
            & 0.569 & 0.534  & 0.437 & 0.593\\
            \midrule
    \multirow{1}{*}{Model 5} 
            & 0.552 & 0.549 & 0.431 & 0.567\\
            \midrule
    \multirow{1}{*}{Model 6} 
            & 1.316 & 1.316 & 0.751 & 1.412\\
\bottomrule 
\end{tabular}}
\caption{Average length of 90\% confidence intervals for the GATE with linear interference (simple linear interference and truncated linear-in-means) outcome models.} 
\label{table:ci-results-linear-length} % A label for referencing this table elsewhere, references are used in text as \ref{label}
\end{table*}

As we can see from the results, our block bootstrap gives us near nominal coverage for ReFeX-LASSO and slightly worse but still close to nominal coverage for post-ReFeX-LASSO. However, the naive bootstrap fails to deliver confidence interval with nominal coverage. In fact, naive bootstrap-based confidence intervals can give us very bad coverage in some cases.  We are also able to get good confidence intervals if we use the oracle clustering that is associated with the network. In scenarios where there are clear natural clusters in the network, these clusters can be a good default choice to use for block bootstrap. Moreover, as is shown in Table~\ref{table:ci-results-linear-length}, both the block bootstrap confidence interval for ReFeX-LASSO and the block bootstrap confidence interval for post-ReFeX-LASSO are of reasonable length.

We conclude this section with a simulation to show why choosing the $k$ for $k$-hop max clustering adaptively in our block bootstrap procedure is important and how partitioning the graph into just two clusters fails to give correct coverage. To this end, we consider using 2-hop max and 3-hop max clustering to divide units into clusters as well as randomly divide units into five clusters, i.i.d., without considering the underlying graph structure. We choose to consider 2-hop max and 3-hop max as we found in the simulations that in most of the cases ReFeX-LASSO will stop after selecting features about 2-hop neighborhoods. For Cai network, on average 2-hop max clustering and 3-hop clustering produce 267 and 269 clusters respectively. We choose to compare them with a five-cluster clustering as five is a lot less than the number of clusters we may have using $k$-hop max clustering. We rerun the block bootstrap procedure with these new clusters for Model 6 using ReFeX-LASSO. Table~\ref{table:ci-results-heu} shows the coverage of the confidence intervals. As we can see, contrast to the 91\% coverage in Tablr~\ref{table:ci-results-linear-coverage} provided by the adaptive $k$-hop max based block bootstrap, all these three clustering methods fail to give us nominal coverage. In particular, completely ignoring the graph structure (``five clusters'') leads to confidence intervals with really poor coverage.

\begin{table*}[ht]
\centering
\scalebox{1.0} {
\begin{tabular}{c c c c c} 
\toprule 
Model & 2-hop max & 3-hop max & Five clusters\\
\midrule
    \multirow{1}{*}{Model 6} 
            & 84\% & 89\% & 45\%\\
\bottomrule 
\end{tabular}}
\caption{Coverage of block bootstrap 90\% confidence intervals for the GATE using different graph clustering algorithms with Model 6 as the true outcome model.} 
\label{table:ci-results-heu} % A label for referencing this table elsewhere, references are used in text as \ref{label}
\end{table*}

\section{Real data example} \label{sec:real-data-example}
In this section, we would like to apply our method to a real experiment where interference is known to exist and simple estimators such as difference-in-means would give poor GATE estimates. We consider data from the intervention in \cite{cai2015}. They designed a randomized experiment to study the role of social networks on insurance adoption in rural China. Specifically, a random subset of farmers were provided with intensive information sessions about the an insurance product. \cite{cai2015} found that the diffusion of insurance knowledge drove network effects in product adoption. Hence, this data is ideal for our purpose in the sense that we know for sure that SUTVA is violated and we should not trust the simple difference-in-means estimate for estimating the GATE. Moreover, though we know that network effects do exist, defining an exact exposure model as in \cite{aronow2017} is difficult. Hence, analysis done in \cite{Chin+2019} is limited since there only four pre-specified features were considered and hence the regression adjustment estimator implicitly assumed a certain exposure model. We revisit this experiment and estimate the GATE using our method. 

In the original field experiment in \cite{cai2015} the intensive information sessions were offered in two separate rounds, leading to four separate treatment arms. For our purpose, following \cite{Chin+2019}, we simplify the experiment by viewing the two intensive information sessions as the same treatment arm. Hence, we reduce the original field experiment to a binary randomized experiment. As in \cite{cai2015}, the outcome variable is set to be the binary
indicator variable for the weather insurance adoption, and we do not include villagers whose treatment or response information was missing as well as villagers whose network information was missing. We also combine all the villages into one social network, denoting this single social network by $G$. In summary, we have 4,382 nodes and 17,069 edges. This network is also the one that we used in Section~\ref{subsec: sim-ci-gate}.

The first step for our method is generating model-free covariates. We use exactly the same set of base features as in the previous simulation section---fraction of treated neighbors, number of treated neighbors, fraction of edges in neighborhood that connects a treated unit and a control unit and also fraction of edges in neighborhood that connects a treated unit and a treated unit. We then use ReFeX-LASSO to generate a group of covariates, using mean and variance aggregation functions (again, as in the previous simulation section) and estimate the GATE by adjusting for these covariates with a linear model. We compare the standard error estimate from block bootstrap with the one computed in \cite{Chin+2019}.
\begin{table}[ht]
\centering
\begin{tabular}{lrr}
  \hline
Estimator & Estimate & Standard Error \\ 
  \hline
DM & 0.078 & ------ \\ 
  H\'ajek\_1hop ($q = 0.75$) & 0.163 & ------ \\ 
  H\'ajek\_2hop ($q = 0.75$) & 0.167 & ------ \\ 
  $\hat{\tau}_{\text{chin}}$ & 0.122 & 0.056 \\ 
  $\hat{\tau}_{\text{num}}$ & 0.178 & 0.027 \\
  $\hat{\tau}_{\text{refex-lasso}}$ & 0.178 & 0.043 \\
   \hline
\end{tabular}
\caption{Estimates and standard errors of different estimators for the global average treatment effect on insurance adoption~\cite{cai2015}.} 
\label{table:cai-est}
\end{table}

Table \ref{table:cai-est} shows the resulting GATE estimates, where $\hat{\tau}_{\text{chin}}$ is the estimator in \cite{Chin+2019} that adjusts for four covariates: the fraction of treated neighbors, the number of treated neighbors, the fraction of treated neighbors in 2-hop neighborhoods, the number of treated neighbors in 2-hop neighborhoods. Meanwhile, $\hat{\tau}_{\text{num}}$ only adjusts for the number of treated neighbors and $\hat{\tau}_{\text{refex-lasso}}$ is the ReFeX-LASSO based adjustment estimator. DM refers to the difference-in-means estimator. H\'ajek\_1hop assumes a fractional exposure model for 1-hop neighborhood while H\'ajek\_2hop assumes a fractional exposure model for 2-hop neighborhood, i.e., we use \eqref{sim-est: hajek} but consider 2-hop neighbors instead. The intuition is that sometimes units that are not direct neighbors but neighbors of direct neighbors matter as well and by considering fractional exposure model for 2-hop neighborhood we are able to take these units into account for the exposure model. We notice that $\hat{\tau}_{\text{num}}$ and $\hat{\tau}_{\text{refex-lasso}}$ give us the same estimate and indeed, the only covariate selected from ReFeX-LASSO is the number of treated neighbors. Compared to $\hat{\tau}_{\text{chin}}$, $\hat{\tau}_{\text{refex-lasso}}$ has smaller standard error and a larger estimate of the effect. Finally, though $\hat{\tau}_{\text{num}}$ and $\hat{\tau}_{\text{refex-lasso}}$ give nearly the same estimates (same up to three decimal digits), we see that the former as a smaller standard error. The reasons are twofold. First, bootstrap in general is conservative. Second, ReFeX estimate should have larger variance as we have a random selection procedure involved.

\section{Discussion}
In this paper, we have developed a method to do estimation and inference for the global average treatment effect (GATE) when network interference is present. We develop a procedure that can be used to estimate the GATE without pre-specifying either exposure mappings or outcome models. We also give a way to construct confidence intervals for the GATE using a block bootstrap. We evaluate our method both through simulations and a real data example.

Many interesting avenues of further investigation have been left unexplored in this manu-script. First, our results only consider designs that satisfy the uniformity assumption (e.g., Bernoulli design): this is, of course, limiting, but it does present a useful benchmark. We are particularly interested in exploring how to extend our work to designs that violate the uniformity assumption such as cluster randomized design. This is challenging since the covariates we adjust for may be correlated with the treatment assignment. Second, while our simulations show that the block bootstrap behaves well in practice, formal results are absent for anything other than a simple toy setting. Third, beyond linear adjustment we may also want to have a completely nonlinear model to estimate the outcomes using the covariates returned from the ReFeX-LASSO feature generation and selection process.

\subsection*{Acknowledgements}
This work was supported in part by ARO
MURI award \#W911NF-20-1-0252 and NSF CAREER Award \#2143176.

%%%%%%%%%%%%%%%%%%%%%%%%%%%%%%%%%%%%%%%%%%%%%%
%% Single Appendix:                         %%
%%%%%%%%%%%%%%%%%%%%%%%%%%%%%%%%%%%%%%%%%%%%%%
%\begin{appendix}
%\section*{???}%% if no title is needed, leave empty \section*{}.
%\end{appendix}
%%%%%%%%%%%%%%%%%%%%%%%%%%%%%%%%%%%%%%%%%%%%%%
%% Multiple Appendixes:                     %%
%%%%%%%%%%%%%%%%%%%%%%%%%%%%%%%%%%%%%%%%%%%%%%
%% if your bibliography is in bibtex format, uncomment commands:
\bibliographystyle{agsm}
\bibliography{submission.bib}       % Bibliography file (usually '*.bib')

\newpage

\begin{appendix}

\section{Proofs} \label{appendix: proofs}
The proofs of Proposition~\ref{prop: selection1} and Proposition~\ref{prop: selection2} will be exactly the same as the proofs of Proposition 1 and Proposition 2 in \cite{sequentialLasso} once we realize that as long as the features that are included in the penalty do not overlap with the features that have already been selected then we can just use the proofs in \cite{sequentialLasso}, i.e., though our sequential selection procedure is different from that in \cite{sequentialLasso}, we share the same properties that make these two propositions hold.
\begin{proof} [Proof of Proposition~\ref{prop: selection1}]
We denote by $X(s)$ the design matrix with features in $s$, i.e., if $|s| = m$ then $X(s)$ is a $n \times m$ matrix. At the $(t+1)-th$ iteration, $\beta$ will be a $(|s_{*t}| + i_{t+1})$-dimensional vector and we denote by $\beta(s)$ the $|s|$-dimensional vector with only coordinates of $\beta$ that are in $s$. Finally, we denote by $A_{t+1}$ the set $\{u_1^{t+1}, u_2^{t+1}, \cdots, u_{i_{t+1}}^{t+1}\}$. 

First we note that since $u_j^{t+1} \in \mathcal{R}(s_{*t})$, $\exists v \in \mathbb{R}^{|s_{*t}|}$ such that $u_j^{t+1} = X(s_{*t})v$. We now consider the objective function $l_{t + 1}$ at the $(t + 1)$-th iteration.
\begin{align*}
    l_{t+1} &= \|y - X(s_{*t})(\beta(s_{*t}) + \beta(\{j\})v) - X(A_{t+1}/\{j\})\beta(A_{t+1}/\{j\})\|_2^2 \\
    & \qquad + \lambda\left(|\beta(\{j\})| + \|\beta(A_{t+1}/\{j\})\|_1\right) \\
    &= \|y - X(s_{*t})\tilde{\beta}(s_{*t}) - X(A_{t+1}/\{j\})\beta(A_{t+1}/\{j\})\|_2^2 \\
    & \qquad + \lambda\left(|\beta(\{j\})| + \|\beta(A_{t+1}/\{j\})\|_1\right) \\
    & \geq  \|y - X(s_{*t})\tilde{\beta}(s_{*t}) - X(A_{t+1}/\{j\})\beta(A_{t+1}/\{j\})\|_2^2 \\
    & \qquad + \lambda\|\beta(A_{t+1}/\{j\})\|_1
\end{align*}
Hence, when $l_{t+1}$ is minimized, $\beta(\{j\})$ must be 0 and $j \notin s_{*(t+1)}$.
\end{proof}

\begin{proof} [Proof of Proposition~\ref{prop: selection2}]
Again we consider the objective function at the $(t+1)$-th iteration. 
\begin{equation*}
    l_{t+1} = \|y - X(s_{*t})\beta(s_{*t}) - X(A_{t+1})\beta(A_{t+1})\|_2^2 + \lambda\|\beta(A_{t+1})\|_1.
\end{equation*}
Differentiating $l_{t+1}$ with respect to $\beta(s_{*t})$, we have
\begin{align*}
    \frac{\partial l_{t+1}}{\partial\beta(s_{*t})} = -2X^T(s_{*t})y + 2X^T(s_{*t})X(s_{*t})\beta(s_{*t}) + 2X^T(s_{*t})X(A_{t+1})\beta(A_{t+1}).
\end{align*}
Setting the above derivative to zero, we have that
\begin{equation}
    \hat{\beta}(s_{*t}) = [X^T(s_{*t})X(s_{*t})]^{-1}X^T(s_{*t})[y - X(A_{t+1})\beta(A_{t+1})]. \label{eq:beta_hat}
\end{equation}
Substituting \eqref{eq:beta_hat} into the objective function, we obtain
\begin{align*}
    l_{t+1} &= \|y - X(s_{*t})\beta(s_{*t}) - X(A_{t+1})\beta(A_{t+1})\|_2^2 + \lambda\|\beta(A_{t+1})\|_1 \\
    &= \|y - X(s_{*t})[X^T(s_{*t})X(s_{*t})]^{-1}X^T(s_{*t})[y - X(A_{t+1})\beta(A_{t+1})] - X(A_{t+1})\beta(A_{t+1})\|_2^2 \\
    & \qquad + \lambda\|\beta(A_{t+1})\|_1 \\
    &= \|(I - X(s_{*t})[X^T(s_{*t})X(s_{*t})]^{-1}X^T(s_{*t}))y \\
    &\qquad - (I - X(s_{*t})[X^T(s_{*t})X(s_{*t})]^{-1}X^T(s_{*t}))X(A_{t+1})\beta(A_{t+1})\|_2^2 \\
    &\qquad + \lambda\|\beta(A_{t+1})\|_1.
\end{align*}
Hence minimizing $l_{t+1}$ does not affect $\hat{\beta}(s_{*t})$ and $\hat{\beta}(s_{*t})$ will be almost surely nonzero.
\end{proof}

Now we show the proof Theorem~\ref{thm: consistency}. We will make use of standard results about LASSO $\ell_2$-error bounds. Recall the following result \cite{wainwright_2019}:
\begin{lemma} \label{thm: lasso-pred-error}
Suppose $y = X\theta^{*} + w$ ($X \in \mathbb{R}^{n \times d}$) and consider the Lagrangian Lasso with a strictly positive regularization parameter $\lambda_n \geq 2\|\frac{\mathbf{X}^Tw}{n}\|_{\infty}$. Suppose further that $\theta^*$ is supported on a subset $S$ of cardinality $s$, and the design matrix satisfies the $(\kappa; 3)$-RE condition over $S$, then
    \begin{equation*}
        \|\hat{\theta} - \theta^{*}\|_2 \leq \frac{3}{\kappa}\sqrt{s}\lambda_n.
    \end{equation*}
\end{lemma}
We can show that if the design matrix is $C-$column normalized, i.e.,
\begin{equation*}
\max_{j = 1, \cdots, d}\frac{\|X_j\|_2}{\sqrt{n}} \leq C,
\end{equation*}
then the choice $\lambda_n = 2C\sigma(\sqrt{\frac{2\log{d}}{n}} + \delta)$ is valid with probability at least $1 - 2e^{-\frac{n\delta^2}{2}}$. We thus proceed with the main proof.
\begin{proof}
Notice that $\|\frac{\mathbf{X}^Tw}{n}\|_{\infty}$ corresponds to the absolute maximum of $d$ zero-mean Gaussian random variables by definition of infinity norm and each with variance at most $\frac{C^2\sigma^2}{n}$. Hence, from the Gaussian tail bound, we then have
\begin{equation*}
    \mathbb{P}\left(\bigg\|\frac{\mathbf{X}^Tw}{n}\bigg\|_{\infty} \geq C\sigma\left(\sqrt{\frac{2\log d}{n}} + \delta\right)\right) \leq 2e^{-\frac{n\delta^2}{2}}.
\end{equation*}
\end{proof}
With this particular choice of $\lambda_n$, the lemma implies the upper bound
\begin{equation}
    \|\hat{\theta} - \theta^{*}\|_2 \leq \frac{6C\sigma}{\kappa}\sqrt{s}\left(\sqrt{\frac{2\log{d}}{n}} + \delta\right) \label{eq: l2-error-bound}
\end{equation}
with the same high probability \cite{wainwright_2019}. 

Now we are ready to prove consistency. First notice that
\begin{align*}
    |\hat{\tau} - \tau| &= \bigg|\frac{1}{n}\sum_{i = 1}^{n}\left[(\hat{\beta}_1 - \beta_1^{*})^Tu_i^{gt} - (\hat{\beta}_0 - \beta_0^{*})^Tu_i^{gc}\right]\bigg| \\ 
    &\leq \frac{1}{n}\sum_{i = 1}^{n}\bigg|\left[(\hat{\beta}_1 - \beta_1^{*})^Tu_i^{gt} - (\hat{\beta}_0 - \beta_0^{*})^Tu_i^{gc}\right]\bigg| \\ 
    &\leq \frac{1}{n}\sum_{i = 1}^{n}(\|\hat{\beta}_1 - \beta_1^{*}\|_2\|u_i^{gt}\|_2 + \|\hat{\beta}_0 - \beta_0^{*}\|_2\|u_i^{gc}\|_2) \\
    &\leq C\sqrt{M}(\|\hat{\beta}_1 - \beta_1^{*}\|_2 + \|\hat{\beta}_0 - \beta_0^{*}\|_2)
\end{align*}
Let $n_0$ be the number of control units and $n_1$ be the number of treated units. Then by strong law of large numbers, $\frac{n_0}{n} \xrightarrow{a.s.} 1-p$ and $\frac{n_1}{n} \xrightarrow{a.s.} p$. Since the design matrices $U^0$ and $U^1$ satisfy the RE condition, both $\|\hat{\beta}_1 - \beta_1^{*}\|_2$ and $\|\hat{\beta}_0 - \beta_0^{*}\|_2$ converge to 0 in probability by the bound \eqref{eq: l2-error-bound}. Thus $\hat{\tau}\xrightarrow{\mathbb{P}} \tau$.

\begin{proof} [Proof of Proposition~\ref{prop: consistency-example}]
%We use the following inequality which is known as McDiarmid's inequality:
%\begin{lemma} \label{lemma: bounded-diff}
%Suppose $(X_1, \cdots, X_n)$ are independent random variables, and let $f : \mathbb{R}^n \rightarrow \mathbb{R}$ satisfy the bounded differences %property with constants $L_1, \cdots, L_n$. Then
%\begin{equation*}
%    \mathbb{P}(|Z - \mathbb{E}[Z]| \geq t) \leq 2\exp\left(-\frac{2t^2}{\sum_{i=1}^nL_i^2}\right),
%\end{equation*}
%where $Z = f(X_1, \cdots, X_n)$.
%\end{lemma}
We show that for the setup in Proposition~\ref{prop: consistency-example}, the design matrices satisfy RE condition with probability going to 1. In our proof, the first column of the design matrix represents the fration of treated neighbors while the second column represents the number of treated neighbors. We introduce one extra notations: for each unit $i$, we denote by $m_i$ the size of the cluster unit $i$ belongs to. We show the proof for the design matrix for control units, $U^0$. Similar proof can be done for $U^1$. After centering, the design matrix we use for estimating $\beta_0$ will be
\begin{equation*}
    \tilde{U}^0 = \begin{bmatrix}
    \frac{1}{n_0}\sum_{i: W_i = 0}(u_i^1 - \bar{u}^1)^2 & \frac{1}{n_0}\sum_{i: W_i = 0}(u_i^1 - \bar{u}^1)(u_i^2 - \bar{u}^2) \\
    \frac{1}{n_0}\sum_{i: W_i = 0}(u_i^1 - \bar{u}^1)(u_i^2 - \bar{u}^2) & \frac{1}{n_0}\sum_{i: W_i = 0}(u_i^2 - \bar{u}^2)^2
    \end{bmatrix}.
\end{equation*}
Here $\bar{u}^1 = \frac{1}{n_0}\sum_{i: W_i = 0}u_i^1$ and $\bar{u}^2 = \frac{1}{n_0}\sum_{i: W_i = 0}u_i^2$. Since the true $\beta_0$ is non-zero only for the first feature, $\mathbb{C}_3(S) = \{\Delta \in \mathbb{R}^2: |\Delta_2| \leq 3|\Delta_1|\}$. For such $\Delta$, we have that
\begin{equation*}
    \frac{1}{n_0}\|\tilde{U}^0\Delta\|_2^2 = \Delta_1^2\frac{1}{n_0}\sum_{i: W_i = 0}(u_i^1 - \bar{u}^1)^2 + 2\Delta_1\Delta_2\frac{1}{n_0}\sum_{i: W_i = 0}(u_i^1 - \bar{u}^1)(u_i^2 - \bar{u}^2) + \Delta_2^2\frac{1}{n_0}\sum_{i: W_i = 0}(u_i^2 - \bar{u}^2)^2
\end{equation*}
Note that since $|\Delta_2| \leq 3|\Delta_1|$, $\Delta_1\Delta_2 \geq -|\Delta_1||\Delta_2| \geq -\frac{1}{3}\Delta_2^2$. Therefore,
\begin{equation}
\begin{split}
    \frac{1}{n_0}\|\tilde{U}^0\Delta\|_2^2 &\geq \frac{1}{n_0}\sum_{i: W_i = 0}(u_i^1 - \bar{u}^1)^2\Delta_1^2 \\
    &+ \left(\frac{1}{n_0}\sum_{i: W_i = 0}(u_i^2 - \bar{u}^2)^2-\frac{1}{3}\frac{1}{n_0}\sum_{i: W_i = 0}(u_i^1 - \bar{u}^1)(u_i^2 - \bar{u}^2)\right)\Delta_2^2. \label{eq:lower_bound}
    \end{split}
\end{equation}
To ease notations, we let \circled{1} = $\frac{1}{n_0}\sum_{i: W_i = 0}(u_i^1 - \bar{u}^1)^2$, \circled{2} = $\frac{1}{n_0}\sum_{i: W_i = 0}(u_i^2 - \bar{u}^2)^2$ and \circled{3} = $\frac{1}{n_0}\sum_{i: W_i = 0}(u_i^1 - \bar{u}^1)(u_i^2 - \bar{u}^2)\Delta_2^2$. Now, we analyze each term separately.
\begin{align*}
    \circled{1} &= \frac{1}{n_0}\sum_{i: W_i = 0}(u_i^1 - \bar{u}^1)^2 \\
    &= \frac{1}{n_0}\sum_{i: W_i = 0}(u_i^1)^2 - (\bar{u}^1)^2 \\
    &= \frac{n}{n_0}\frac{1}{n}\sum_{i=1}^{n}(1 - W_i)(u_i^1)^ 2 - (\bar{u}^1)^2 \\
    &= \frac{n}{n_0}\frac{1}{n}\sum_{i=1}^{n}(1 - W_i)(u_i^1)^2 -  \left(\frac{n}{n_0}\frac{1}{n}\sum_{i = 1}^n(1 - W_i)u_i^1\right)^2.
\end{align*}
Consider the random variables $\{(1-W_i)(u_i^1)^2\}_{i=1}^{n}$ and $\{(1-W_i)u_i^1\}_{i=1}^{n}$. Since we have disjoint clusters and the number of units in each cluster is bounded by $M$, the sum of covariance term is at most $O(n)$ and hence weak law of large numbers applies for both sequences. Therefore,
\begin{equation*}
    \frac{1}{n}\sum_{i=1}^{n}(1 - W_i)(u_i^1)^2 - \left[p(1-p)^2\frac{1}{n}\sum_{i=1}^{n}\frac{1}{m_i-1} + p^2(1-p)\right] \xrightarrow{\mathbb{P}} 0.
\end{equation*}
Similarly,
\begin{equation*}
    \frac{1}{n}\sum_{i=1}^{n}(1 - W_i)u_i^1 \xrightarrow{\mathbb{P}} p(1-p).
\end{equation*}
Note that $n/n_0 \xrightarrow{\mathbb{P}} 1/(1-p)$, we obtain
\begin{equation*}
     \circled{1} - \left[p(1-p)\frac{1}{n}\sum_{i=1}^{n}\frac{1}{m_i-1}\right] \xrightarrow{\mathbb{P}} 0.
\end{equation*}
Here \circled{2} can be done similarly:
\begin{equation*}
     \circled{2} - \left[p(1-p)\frac{1}{n}\sum_{i=1}^n(m_i-1) + p^2\frac{1}{n}\sum_{i=1}^n(m_i-1)^2 - p^2\left(\frac{1}{n}\sum_{i=1}^n(m_i-1)\right)\right] \xrightarrow{\mathbb{P}} 0.
\end{equation*}
For \circled{3}, we have that
\begin{align*}
    \circled{3} &= \frac{1}{n_0}\sum_{i: W_i = 0}(u_i^1 - \bar{u}^1)(u_i^2 - \bar{u}^2) \\
    &= \frac{1}{n_0}\sum_{i: W_i = 0}u_i^1u_i^2 - \bar{u}^1\bar{u}^2.
\end{align*}
Notice that we have already shown that
\begin{equation*}
    \bar{u}^1 \xrightarrow{\mathbb{P}} p, \qquad \bar{u}^2 \xrightarrow{\mathbb{P}} p\frac{1}{n}\sum_{i=1}^n(m_i-1).
\end{equation*}
Hence, $\bar{u}^1\bar{u}^2 \xrightarrow{\mathbb{P}} p^2\frac{1}{n}\sum_{i=1}^n(m_i-1)$. Moreover,
\begin{align*}
    \frac{1}{n_0}\sum_{i: W_i = 0}u_i^1u_i^2 &= \frac{n}{n_0}\frac{1}{n}\sum_{i=1}^{n}(1 - W_i)u_i^1u_i^2.
\end{align*}
Again by weak law of large numbers,
\begin{equation*}
    \frac{1}{n}\sum_{i=1}^{n}(1 - W_i)u_i^1u_i^2 - \left[p(1-p)^2 + p^2(1-p)\frac{1}{n}\sum_{i=1}^n(m_i-1)\right] \xrightarrow{\mathbb{P}} 0.
\end{equation*}
Hence, $\circled{3} - \left[p(1-p) + p^2\frac{1}{n}\sum_{i=1}^n(m_i-1)\right] \xrightarrow{\mathbb{P}} 0$. Put all these pieces together, we obtain
\begin{equation*}
\begin{split}
    \text{RHS of } \eqref{eq:lower_bound} &- \Biggl\{\left[p(1-p)\frac{1}{n}\sum_{i=1}^{n}\frac{1}{m_i-1}\right]\Delta_1^2 \biggl.\\ 
    &+ \left[p(1-p)\frac{1}{n}\sum_{i=1}^n(m_i-1) + p^2\frac{1}{n}\sum_{i=1}^n(m_i-1)^2 - p^2\left(\frac{1}{n}\sum_{i=1}^n(m_i-1)\right)\right. \\
    &- \left.\frac{1}{3}\left(p(1-p) + p^2\frac{1}{n}\sum_{i=1}^n(m_i-1)\right)\right]\Delta_2^2 \Biggr\} \xrightarrow{\mathbb{P}} 0.
\end{split}
\end{equation*}
Notice that $m_i \geq 3$ and $m_i \leq M$ for each $i$, we conclude that for $\kappa = \min\{\frac{p(1-p)}{M-1}, \frac{5}{3}p - \frac{1}{3}p^2\}$,
\begin{equation*}
    \frac{1}{n_0}\|\tilde{U}^0\Delta\|_2^2 \geq \kappa \|\Delta\|_2^2 \quad \text{w.p.} \quad \rightarrow 1.
\end{equation*}
\end{proof}
%To finish the proof, it then suffices to show that with Bernoulli design and our assumptions on the features, the design matrices satisfy the $(\kappa; 3)$-RE condition over $S_{*}$ with probability tending to 1. We use $U^{0}$ and $U^{1}$ to denote the design matrices for the two linear regressions we run for estimating $\beta_0$ and $\beta_1$ respectively, i.e. the rows of $U^0$ are corresponding $M-$dimensional feature vectors of control units and the rows of $U^1$ are corresponding $M-$dimensional feature vectors of treated units. Now, to run LASSO, we standardize $U_1$ and $U_0$ so that each column is zero-mean and unit-variance. Hence, we need to show that the standardized design matrix $\tilde{U}_1$ satisfies the $(\kappa; 3)$-RE condition over $S_{*}$ with probability tending to 1. We relabel the units such the first $n_1$ units are treated units. 
\section{Supplementary Materials} \label{appx: supp}
\begin{definition}[The nonlinear model in simulations]
Suppose the assignment vector is $w$, then for each unit $i$, the response is
\begin{equation*}
y_i(w) = -5 + 2z_iw_i + 0.03\nu_i + \frac{1}{1 + 0.001\exp{(-0.03\nu_i + 9)}} + \frac{10}{3 + \exp{(-8\rho_i + 3.2)}} + \epsilon_i.
\end{equation*}
Here, $z_i, \epsilon_i \overset{\mathrm{i.i.d}}{\sim} \mathcal{N}(0,1)$, $\rho_i$ is the fraction of treated neighbors for unit $i$ and $\nu_i$ is the number of treated neighbors for unit $i$.
\end{definition}

\end{appendix}

\end{document}